\documentclass[10pt, journal]{IEEEtran}
\IEEEoverridecommandlockouts

\usepackage{amsmath}
\usepackage{amssymb}
\usepackage{amsthm}
\usepackage{mathtools}
\usepackage{enumerate}
\usepackage{bbm}
\usepackage{algorithm}
\usepackage{algorithmicx}
\usepackage[noend]{algcompatible}
\usepackage{epsfig}
\usepackage{color}
\usepackage{comment}
\usepackage[dvipsnames]{xcolor}
\usepackage[normalem]{ulem}
\usepackage{subfigure}

\usepackage{caption}
\usepackage{cases}
\usepackage{url}
\usepackage{cite}
\usepackage{mleftright}
\usepackage{xparse}
\usepackage{multirow}
\usepackage{wasysym}
\usepackage{graphicx}
\usepackage{etoolbox}

\newcommand{\GREEN}[1]{{\color{PineGreen}{#1}}}

\newtheorem{proposition}{Proposition}
\newtheorem{lemma}{Lemma} 

\DeclareMathOperator*{\argmax}{arg\,max}

\newcommand{\noLink}{|\mathcal{L}|}
\newcommand{\setLink}{\mathcal{L}}
\newcommand{\setNode}{\mathcal{V}}

\newcommand{\setS}{\mathcal{S}}
\newcommand{\vlambda}{\boldsymbol{\lambda}}
\newcommand{\vmu}{\boldsymbol{\mu}}

\newcommand{\vq}{\boldsymbol{q}}
\newcommand{\maxq}{q^*}

\newcommand{\weight}{W}
\newcommand{\mweight}{w}
\newcommand{\sweight}{\hat{w}}
\newcommand{\vmweight}{\boldsymbol{w}}
\newcommand{\iweight}{\bar{w}}
\newcommand{\viweight}{\bar{\boldsymbol{w}}}

\newcommand{\totalmw}{r}
\newcommand{\maxtotalmw}{r^*}

\newcommand{\setOpts}{\mathcal{S}^{*}}
\newcommand{\opts}{S^*}
\newcommand{\nopts}{\mathcal{S}^\alpha}
\newcommand{\cnopts}{\bar{{\mathcal{S}}}^\alpha}
\newcommand{\snopts}{\overline{\mathcal{S}}}
\newcommand{\alpR}{Reg^{\alpha}}

\newcommand{\framel}{T}
\newcommand{\alpk}{\frac{k-1}{k+1}}
\newcommand{\nochosen}{\hat{\tau}}

\newcommand{\ucbI}{\iweight}

\newcommand{\augstar}{\mathcal{A}^*}
\newcommand{\alpgapmin}{\Delta_{\min}^{\alpha}}
\newcommand{\alpgapmax}{\Delta_{\max}^{\alpha}}
\newcommand{\orderF}{n}

\newcommand{\lya}{L}
\newcommand{\dri}{D}
\newcommand{\arv}{a}
\newcommand{\serv}{X}

\newcommand{\mserv}{\mu}

\newcommand{\schedule}{\mathbf{S}}
\newcommand{\dep}{X}
\newcommand{\queue}{q}
\newcommand{\sregion}{\Lambda}

\newtoggle{tech_report}
\toggletrue{tech_report}
\begin{document}

%\title{Distributed Link Scheduling with Unknown Link Rates in Multi-hop Wireless Networks}
%\title{Distributed Learning and Scheduling in Multi-hop Wireless Networks}
\title{A Learning-based Distributed Algorithm for Scheduling in Multi-hop Wireless Networks}

\author{Daehyun Park, Sunjung Kang,~\IEEEmembership{Student~Member,~IEEE}, and Changhee Joo,~\IEEEmembership{Senior~Member,~IEEE}%
\IEEEcompsocitemizethanks{\IEEEcompsocthanksitem D. Park was with ECE, UNIST, Ulsan, South Korea.\protect\\
Email: eoy13@unist.ac.kr
\IEEEcompsocthanksitem S. Kang is with ECE, OSU, Columbus, Ohio.\protect\\
Email: kang.853@osu.edu
\IEEEcompsocthanksitem C. Joo is the corresponding author, and he is with CSE, Korea University, Seoul, South Korea.\protect\\
Email: changhee@korea.ac.kr}
\thanks{This work was supported in part by the NRF grant funded by the Korea government (MSIT) (No. NRF-2021R1A2C2013065), and in part by the MSIT, Korea, under the ICT Creative Consilience program (IITP-2020-0-01819) supervised by the IITP.}
%\thanks{Manuscript received ...; revised ....}
}

%{Shell \MakeLowercase{\textit{et al.}}: Bare Demo of IEEEtran.cls for Computer Society Journals}

\IEEEtitleabstractindextext{%
\begin{abstract}
We address the joint problem of learning and scheduling in multi-hop wireless network without a prior knowledge on link rates. Previous scheduling algorithms need the link rate information, and learning algorithms often require a centralized entity and polynomial complexity. These become a major obstacle to develop an efficient learning-based distributed scheme for resource allocation in large-scale multi-hop networks. In this work, by incorporating with learning algorithm, we develop provably efficient scheduling scheme under packet arrival dynamics without a priori link rate information. We extend the results to distributed implementation and evaluation their performance through simulations. 
%we incorporate a low-complexity scheduling algorithm into the learning procedure, and develop new distributed schemes for scheduling and learning. We analytically evaluate distributed the schemes and show that they achieve the time-order optimal performance in learning and the throughput performance arbitrarily close to the optimal. To our best knowledge, they are the first schemes that achieve probably efficient regret and throughput performance, and that are amenable to implement in a distributed manner. We verify our results through simulations. 
\end{abstract}

}

\maketitle

\IEEEdisplaynontitleabstractindextext

% \fi
\IEEEpeerreviewmaketitle

\section{Introduction} \label{sec:intro}

\IEEEPARstart{A}{s}
one of the key functions in wireless communication networks, link scheduling determines which links should be activated at what time. The problem is challenging, in particular, in multi-hop\footnote{In this work, we use the term `multi-hop' for the arbitrary interference relationship between wireless links, and consider only single-hop traffic flows, i.e., a packet departs the network after a single transmission. If a scheduling scheme is suboptimal with these single-hop flows over multi-hop network, then clearly it is suboptimal with multi-hop flows. This approach has been widely adopted to investigate the performance of scheduling schemes without being affected by the functions of other layers such as routing and congestion control~\cite{joo_jwcn10, lin_ton06, joo_tac09, joo_tmc12a, wu07, lin_ton09, joo_ton09a}.}  networks due to non-linear interference relationship between wireless links. The seminal work of Tassiulas and Ephremides has shown that the Maximum Weighted Matching (MWM) algorithm that maximizes the queue weighted rate sum can achieve the optimal throughput under packet arrival dynamics~\cite{tassiulas_tac92}. Due to high computational complexity of MWM~\cite{joo_jwcn10}, alternative low-complexity scheduling solutions with comparable performance such as Greedy Maximal Matching (GMM) or Longest Queue First (LQF) have attracted much attention~\cite{lin_ton06, joo_tac09}. However, since GMM still has linear complexity that increases with the network size, it can be hardly used in large-size multi-hop networks.

There has been extensive research on developing efficient scheduling algorithms that have sublinear complexity, yet perform provably well in multi-hop wireless networks. An approximation to GMM with logarithmic complexity has been developed in~\cite{joo_tmc12a}. Random access technique with explicit neighborhood information exchanges has been explored at some expense of performance~\cite{wu07, lin_ton09, joo_ton09a, jinho_jcn20}. Several studies have shown that the optimal throughput performance is achievable, either by taking the pick-and-compare approach~\cite{hajek88, bui_ton09}, or by exploiting the carrier-sensing functionality~\cite{jiang_ton10, ni_ton12}.
There have been also attempts to develop provably efficient scheduling algorithms that work with time-varying wireless channels~\cite{joo_tmc12a, joo_tmc13} or with complex SINR interference model~\cite{borbash06b, jin_ton14}.
The aforementioned scheduling schemes provide performance guarantees under packet arrival dynamics. However, they are limited to deterministic link rates that are \emph{known a priori} or at the time of scheduling. Their extension to the case when the link rates are unknown is not straightforward.

In this work, we consider the scheduling problem, where link rates and statistics are unknown a priori. This occurs when new applications try to operate efficiently under uncertainty caused by wireless fading, interference, limited feedback, measurement error, system dynamics, etc~\cite{li_wc20, zhao_jsac2007, stahlbuhk_adhoc2019}. 
We assume that an instance link rate is revealed when it is accessed/scheduled, and it is drawn from an unknown static distribution. Our goal is to find an appropriate sequence of link schedules that maximize throughput under packet arrival dynamics, while quickly learning the link rates and queue states.
%Our goal is to find a sequence of non-interfering link sets (i.e., feasible schedules) in multi-hop wireless networks to maximize throughput performance while learning link rates.

Focusing on the learning aspect, the problem can be viewed as a variant of Multi-Armed Bandit (MAB) problems, in which one repetitively plays a set of arms to maximize the reward sum~\cite{chen_icml2013}. 
The performance of a learning algorithm is often evaluated by \emph{regret}, which is the difference in the total expected reward obtained by an optimal policy and that by the learning algorithm. Lai and Robbins have shown that the regret grows at least at logarithmic rate of time~\cite{lai_1985}, and several index-type learning algorithms with the order-optimal regret have been developed~\cite{auer2002, anantharam_tac1987}.
For a large-scale multi-hop wireless network, it is imperative to design algorithms that are amenable to implement in a distributed manner. In~\cite{liu_tsp2010}, the authors have developed a distributed learning algorithm that selects best $M$ out of $N$ arms, where each of $M$ users selects an arm taking into consideration mutual collision. By employing a time-division selection approach, the scheme is shown to achieve logarithmic regret. 
Chen et al.~\cite{chen_icml2013} and Gai et al.~\cite{gai_ton2012} have considered more general problems of combinatorial MAB (CMAB) with arbitrary constraint. They have employed an $(\alpha,\beta)$-approximation oracle that can achieve $\alpha$ fraction of the optimal value with probability $\beta$, and developed learning schemes that can achieve the logarithmic growth for $\alpha\beta$ fraction of the optimal expected regret (denoted by \emph{$\alpha\beta$-regret}). However, an oracle with good $\alpha\beta$ (i.e., close to $1$) often has a high-order polynomial complexity, and thus as the network scales, it is not clear whether the scheme is amenable to implement in a distributed manner.

Applying learning algorithms to scheduling in wireless networks, the works of~\cite{gai_globecom2011, anandkumar_jsac2011} have addressed the regret-minimization problem in cognitive radio network settings, and developed distributed schemes with logarithmic regret through prioritized ranking and adaptive randomization.
The authors of~\cite{tibrewal_infocom19} have developed fully distributed schemes that can achieve the logarithmic regret without any information exchange. 
%The authors of~\cite{ren_tvt20} have considered the joint scheduling and learning problem under the SINR-based interference model, taking into account full-duplex transmission. 
%
In~\cite{kang_tmc20}, the authors have successfully lowered the algorithmic complexity to $O(1)$ while achieving the logarithmic regret performance.
Although these aforementioned learning algorithms are amenable to distributed implementation, they are limited to single-hop networks, such as wireless access networks, and to saturated traffic scenarios (i.e., links always have packets to send), and thus cannot accordingly respond to packet arrival dynamics.
Recently, Stahlbuhk et al.~\cite{stahlbuhk_adhoc2019} has developed a joint learning and scheduling scheme that provides provable efficiency under packet arrival dynamics, by incorporating GMM scheduling algorithm and UCB-based learning algorithm. Albeit interesting, it achieves only $\frac{1}{2}$ of the capacity region and has linear complexity, which makes it less attractive for large-size networks.
%Recent work of Stahlbuhk et al.~\cite{stahlbuhk_adhoc2019} is the most related to ours. They have incorporated GMM, which is $(\frac{1}{2},1)$-approximation oracle with linear complexity, for both learning and scheduling in multi-hop wireless networks with unknown link rates, and shown that the GMM-based scheme achieves the logarithmic growth for $\frac{1}{2}$-regret. Albeit interesting, its stability region is limited to $\frac{1}{2}$ of the capacity region, and has linear complexity that makes it less attractive for large-size networks.

In this work, we consider the joint problem of learning and scheduling in multi-hop wireless networks, and develop \emph{low-complexity schemes that achieves near-full capacity region under packet arrival dynamics}. Our contribution can be summarized as follows: 
\begin{itemize}
    \item We develop a joint learning and scheduling scheme with $O(k)$ computational complexity, by successfully incorporating a graph augmentation algorithm with the UCB index. Parameter $k$ is settable for a certain level of performance, in which case the complexity becomes $O(1)$.
    %\item We show that the augmentation algorithm, which can be considered as an $(\alpha,\beta)$-approximation oracle, achieves the rate-optimal logarithmic growth of $\alpha$-regret, for any $\beta>0$.
    \item We show that A$^k$-UCB is an $(\alpha,\beta)$-approximation oracle with $\alpha = \alpk$ and some small non-zero~$\beta$, but can indeed achieve the logarithmic growth for $\alpk$-regret, regardless of the value of $\beta$, which is in contrast to $\alpha\beta$-regret shown in~\cite{chen_icml2013, gai_ton2012}.
    %$\alpk$-fraction of the capacity region at $O(k)$ computational complexity. 
    \item By using the frame structure, we show that A$^k$-UCB achieves $\alpk$ fraction of the capacity region under packet arrival dynamics.
    \item We extend our result and develop $d$A$^k$-UCB scheme that is amenable to implement in a completely distributed manner.
\end{itemize}

The rest of paper is organized as follows. Section~\ref{sec:model} describes our system model. We propose a joint scheme of learning and scheduling in Section~\ref{sec:algorithm}, and analytically evaluate its performance in Section~\ref{sec:evaluation}. We further extend our scheme for distributed implementation in Section~\ref{sec:distributed}. Finally, we numerically evaluate our schemes in Section~\ref{sec:sim} and conclude our work in Section~\ref{sec:conclusion}.

%%%%%%%%%%%%%%%%%%%%%%%%%%
% Network Model
%%%%%%%%%%%%%%%%%%%%%%%%%%
\section{System Model} \label{sec:model}
We consider a multi-hop wireless network denoted by graph $\mathcal{G}=(\setNode,\setLink)$ with the set $\setNode$ of nodes and the set $\setLink$ of directional links. We assume that the connectivity is reciprocal, i.e., if $(u,v) \in \setLink$, then $(v,u) \in \setLink$. A set of links that can be scheduled at the same time is constrained by the primary interference model, under which any node $v$ (either transmitter or receiver) in the network can communicate with at most one of its neighbor nodes $\mathcal{N}(v)$, where $\mathcal{N}(v)=\{ u\in\setNode~|~ (v,u)\in \setLink \}$. Slightly abusing the notation, we also denote the set of links that is connected to $v$ by $\mathcal{N}(v)$.
The primary interference model can represent Bluetooth or FH-CDMA networks as well as capture the essential feature of wireless interference~\cite{lin_ton06, stahlbuhk_adhoc2019}, and has been adopted in many studies on wireless scheduling, e.g., see~\cite{lin_ton06, joo_tac09, joo_tmc12a, wu07, lin_ton09, joo_ton09a} for more detailed description.
Time is slotted, which can be achieved by being equipped with high accuracy GPS. At each time slot, a set of links that satisfies the interference constraints can be simultaneously activated. Such a set of links is called a \emph{matching} (or a feasible schedule), and let $\setS$ denote the set of all matchings. 

At each link $i \in \setLink$, we assume packets arrive following a Bernoulli process with probability $\lambda_i$ (e.g., see~\cite{bui_ton09}). Let $\vlambda$ denote its vector and $\arv_i(t) \in \{0,1\}$ denote the number of arrived packets in time slot $t$. We have $\mathbb{E} [\arv_i(t)] = \lambda_i$.
We assume that the rate of link $i$ is time-varying due to multi-path fading and unknown interference as in~\cite{joo_tmc13}, and it is independently drawn from a (possibly different) distribution with mean $\mserv_i$. Let $\vmu$ denote its vector and $\dep_i(t)\in [0,1]$ denote the instance rate of link~$i$ when it is activated at time slot $t$, with $\mathbb{E} [\dep_i(t)] = \mu_i$. The extension to multiple packet arrivals and departures is straightforward. We assume that $\vlambda$ and $\vmu$ are unknown. 

At time slot $t$, if a policy activates matching $\schedule_t$, then each link $i \in \schedule_t$ accesses the medium and transmits $\dep_{i}(t)$ packets\footnote{or transmits $1$ packet with success probability $\dep_i (t)$.} during the time slot.
Each link~$i$ is associated with an unbounded buffer that queues up packets for transmission. Let $\queue_i(t)$ denote the queue length at link $i$ at the beginning of time slot $t$, which evolves as 
\begin{equation} \label{eq:q_evolv}
    \queue_i(t+1) = 
    \begin{cases}
        \left[ \queue_i(t) - \dep_i(t) \right]^{+} + \arv_i(t), &\text{ if } i \in \schedule_t, \\
        \queue_i(t) + \arv_i(t), &\text{ if } i \notin \schedule_t,
    \end{cases}
\end{equation}
where $[\cdot]^{+} = \max\{ \cdot, 0 \}$. Let $\vq(t)$ denote its vector, and let $\maxq(t) = \max_{i \in \setLink} \queue_i(t)$ denote the maximum queue length in the network at time slot $t$.

We consider a frame structure where each frame has length of $\framel$ time slots. Frame $\orderF$ begins at time slot $t_\orderF=(\orderF-1) \framel + 1$. During frame $\orderF$, i.e., for time slots $t \in [t_\orderF, t_{\orderF+1})$, we define weight $\weight_i(t)$ and its mean $\mweight_i$ of link $i$, respectively, as
\begin{equation} \label{eq:weight}
\textstyle
    \weight_i(t) = \frac{q_i(t_\orderF)}{\maxq(t_\orderF)} \serv_{i}(t),
    ~\text{and}~ 
    \mweight_i = \frac{q_i(t_\orderF)}{\maxq(t_\orderF)}\mu_{i}.
\end{equation}
For $\maxq(t_\orderF)=0$, we define $\weight_i(t)=\serv_{i}(t)$ and $\mweight_i=\mu_i$. Let $\vmweight$ denote its vector $(\mweight_1,\mweight_2,\cdots,\mweight_{\noLink})$, where $|\cdot|$ is the cardinality of the set. We denote the link weight sum of matching $S$ by
\begin{equation}
    \textstyle \totalmw_{\vmweight}(S) = \sum_{i \in S} \mweight_i.
\end{equation} 
For convenience, we let $\maxtotalmw_{\vmweight}=\max_{S \in \setS }~ \totalmw_{\vmweight}(S)$ denote the largest weight sum over all matchings, and we also denote a set of optimal matchings by $\setOpts_{\vmweight} =\argmax_{ S \in \setS} \totalmw_{\vmweight}(S)$. 
For $\alpha \in (0,1]$, we define a set of near-optimal matchings with respect to vector $\vmweight$ as 
\begin{equation} \label{eq:near-optimal}
    \nopts_{\vmweight} = \{S \in \setS ~\vert ~ \totalmw_{\vmweight}(S) \geq \alpha \cdot \maxtotalmw_{\vmweight} \}.
\end{equation}
and define its complement as $\cnopts_{\vmweight} = \setS - \nopts_{\vmweight}$.

In the CMAB framework, a link corresponds to an arm, a matching to a super arm, and the instance link rate to the reward of the link, respectively. We use the terms interchangeably. Note that the regret is defined as the accumulated expected difference between the reward sum associated with an optimal matching and that obtained by the MAB algorithm. Similar to~\cite{chen_icml2013}, we define \emph{$\alpha$-regret} as, for some $\alpha \in (0,1]$, 
\begin{equation} \label{eq:a-regret}
    \alpR(t) = t \cdot \alpha \cdot \maxtotalmw_{\vmweight}
                \textstyle -\mathbb{E}\left[\sum_{\tau=1}^t r_{\vmweight}(\schedule_{\tau})\right],
\end{equation}
which evaluates the performance of an CMAB task at time~$t$.

In the viewpoint of resource allocation, achieving a high reward sum is equivalent to achieving a larger queue-weighted link rate sum, which implies that the links with high demands and high service rates are scheduled first, and thus tends to stabilize the network.
A network is said to be \emph{stable} if the queues of all links are \emph{rate stable}, i.e., $\lim_{t\rightarrow\infty} \frac{q_i(t)}{t} = 0$ with probability $1$ for all $i\in\setLink$. Let $\sregion$ denote the \emph{capacity region}, which is the set of arrival rate vectors $\vlambda$ such that for any $\vlambda \in \sregion$, there exists a policy that can make the network stable. We say that a scheduling policy has the \emph{stability region} $\gamma \Lambda$ for some $\gamma \in [0,1]$, if it can stabilize the networks for any arrival $\vlambda \in \gamma\sregion$.

We aim to develop a joint scheme of learning and scheduling that determines $\schedule_t$ to
\[
\begin{array}{ll}
    \text{\bf maximize}  & \gamma \\
    \text{\bf subject to} 
        & \underset{t\rightarrow\infty}{\lim} \frac{q_i(t)}{t} = 0, ~\text{for any}~ \vlambda \in \gamma\sregion, \\
        & \serv_i(t) \stackrel{iid}{\sim} \mathcal{D}_i, ~\text{and Eq.~(\ref{eq:q_evolv})},
\end{array}
\]
where $\mathcal{D}_i$ denotes the distribution with finite support $[0,1]$. The arrival vector $\vlambda$ and the distributions $\{\mathcal{D}_i\}$ are \emph{unknown} to the controller, and an observation on $\serv_i(t)$ is available only by scheduling link $i$ at time $t$.

%%%%%%%%%%%%%%%%%%%%%%%%%%
% 
%%%%%%%%%%%%%%%%%%%%%%%%%%
\section{Augmentation with UCB index (A$^k$-UCB)} \label{sec:algorithm}

We develop a provably efficient joint scheme of learning and scheduling, by incorporating the augmentation algorithm presented in~\cite{bui_ton09} with UCB index. We first describe the overall algorithm, and then explain the detailed operations.

%We can consider the problem of link-rate learning as a CMAB problem with \emph{linear reward}, where the goal is to maximize the total (weighted) reward sum. It is of \emph{semi-bandit} type, i.e., individual reward of each played arm is revealed. For this class of learning problems, the well-known UCB index~\cite{auer2002} is shown to be very effective.

We consider each frame time as an independent learning period, which allows us to decouple the learning from the scheduling. In the following, we describe the operation of our algorithm during a frame time. For the ease of exposition, we assume that our algorithm runs for time slot $[1,\framel]$, where $\framel$ is the frame length.
%when a new frame starts, all the learning parameters are reset, and new learning period starts.  
%But it also incurs a sort of inefficiency in the learning, which we will address later. 

We start with some notations. 
Let $q_i=q_i(1)$ denote the initial (i.e., at the beginning of the frame) queue length of link (arm) $i$, and let $\maxq=\max_i q_i$. Let $\nochosen_i(t)$ denote the number of times that arm $i$ is played up to time slot $t$, and let $\nochosen_S(t)$ denote the number of times that matching $S$ is played. The UCB index of arm $i$~\cite{auer2002} is defined as 
\begin{equation} \label{eq:index}
    \textstyle 
    %\iweight_{i,t} = \sweight_{i,\nochosen_i(t-1)} + \sqrt{\frac{(L+1)\ln t}{\nochosen_i(t-1)}},
    \iweight_{i,t} = \sweight_i (t-1) + \sqrt{\frac{(\noLink+1)\ln t}{\nochosen_i(t-1)}},
\end{equation}
where $\sweight_i (t) = \frac{q_i}{\maxq} \cdot \frac{1}{\nochosen_i(t)}\sum_{n=1}^{t}  \serv_{i}(n) \cdot  \mathbb{I}\{i \in S_n\}$ denotes average reward of arm $i$ at time slot $t$ weighted by $\frac{q_i}{\maxq}$, and $\mathbb{I}\{ e \} \in \{0,1\}$ denotes the indicator function that equals $1$ if event $e$ occurs, and $0$ otherwise.
All the variables of $\iweight_{i,t},\sweight_i,\nochosen_i,q_i,\maxq$ will be reset at the beginning of each frame.
Let $\viweight_t = (\iweight_{1,t}, \iweight_{2,t}, \dots, \iweight_{\noLink,t})$ denotes the UCB index vector. 
Then $\totalmw_{\viweight_t}(S)$ and $\maxtotalmw_{\viweight_t}$ are the index sum over links in matching $S$ and its maximum value over all possible matchings, respectively. 
Also, we denote $\setOpts_{\viweight_t}$ and $\nopts_{\viweight_t}$ as the set of matchings that achieve $\maxtotalmw_{\viweight_t}$ and those that achieve at least $\alpha \maxtotalmw_{\viweight_t}$, respectively.

\begin{algorithm}[t]
\caption{Frame-based joint learning and scheduling.}\label{alg:generic_ucb}
\begin{algorithmic}[1] 
\STATEx /* Repeat at each frame time */
\STATE Obtain queue constant $q_i$ and $\maxq$ 
\STATE Initialize $\sweight_{i}$ and $\nochosen_i$
\FOR{ $t=1$ to $\noLink$ }
    \STATE Schedule arbitrary matching $\schedule_t$ that has link $t$
    \STATE Update $\sweight_{i}$ and $\nochosen_i$ for each link $i \in \schedule_t$
\ENDFOR
\FOR{ $t={\noLink+1}$ to $\framel$ }
    \STATE Compute UCB index 
        $\ucbI_{i,t} \leftarrow \sweight_{i} + \sqrt{\frac{(\noLink+1)\ln t}{\nochosen_i}}$ 
    \STATE $\langle$  Select matching $\schedule_t$ using $\viweight$ $\rangle$
    \STATE Schedule $\schedule_t$
    \STATE Update $\sweight_{i}$ and $\nochosen_i$ for each link $i \in \schedule_t$
\ENDFOR
%\ENDFOR
\end{algorithmic}
\end{algorithm} 

We develop our joint learning and scheduling algorithm based on the generic UCB index, as shown in Algorithm~\ref{alg:generic_ucb}, where we omit subscript $t$ of $\sweight_i (t)$ and $\nochosen_i(t)$ for brevity. 
At time slot $t=1$, it obtains two constant weight parameters $q_i$ and $\maxq$ (line 1), and schedules arbitrary matchings for $\noLink$ time slots such that each link can be scheduled at least once (lines 3-5). Afterwards, it computes the index of each arm (line 7), selects matching $\schedule_t$ using the indices (line 8), and schedules it (lines 9-10).

%%%%$S_t \in \nopts_{\viweight_t}$ for some predetermined $\alpha \in [0,1]$. 
%If an $\alpha$-oracle algorithm that returns a matching $\schedule_t \in \nopts_{\viweight_t}$ \emph{at every time slot $t$} is used in line 8, then we may achieve $\alpha$ fraction of the capacity region, such as the greedy algorithm with $\alpha = \frac{1}{2}$~\cite{stahlbuhk_adhoc2019}. % and requires a centralized entity. 
%
%If the matching with the \emph{highest} index-sum is selected, i.e., $\alpha=1$, then the algorithm is equivalent to the original UCB learning algorithm, and it is known that an oracle with $\alpha = 1$ has to solve the Maximum Weighted Independent Set (MWIS) problem, which is in general NP-Complete~\cite{joo_jwcn10}.

The key part of the algorithm is about how to select matching $\schedule_t$ in line 8 such that it achieves high performance at low complexity. It has been shown in~\cite{gai_ton2012, chen_icml2013} that, if we use an $(\alpha,\beta)$-approximation oracle in the matching selection, the joint algorithm has $\alpha\beta$-regret performance, which leads to achieving $\alpha\beta$ fraction of the capacity region. The problem is, that the parameter $\beta$ can be very small in particular when the network size is large. To this end, we introduce a class of augmentation algorithms with parameter $k$, which is an $(\alpha,\beta)$-approximation oracle with $\alpha = \alpk$ and some small\footnote{It depends on the network topology and the setting of $p$. In general, we have a smaller $\beta$ for a larger network.} $\beta > 0$, and has $O(k)$ complexity.

%\PURPLE{Our main contribution is to show that by adopting the augmentation algorithm, our joint learning and scheduling algorithm, denoted by A$^k$-UCB, has $\alpha$-regret performance regardless how small $\beta$ is, and thus achieves $\alpk$ fraction of the capacity region. This is not straightforward since the previous analysis on the throughput performance of the augmentation algorithm~\cite{bui_ton09} cannot be used due to inaccurate link rate information. Also considering A$^k$-UCB as an $(\alpha,\beta)$-oracle is not much helpful due to the fact that probability $\beta$ can be arbitrarily small as the network scales up.}

%In the following section, we show that, with appropriate settings, A$^k$-UCB achieves $\alpk$ fraction of the optimal expected reward, regardless of the network size, and thus achieves $\alpk$ fraction of the capacity region, which can be arbitrarily close to $1$ by increasing~$k$.

%\subsection{Augmentation algorithm} \label{Augmentation}

\iftoggle{tech_report}
{
    \begin{algorithm}[t]
    \caption{Augmentation for each node $v$.} \label{alg:scheme}
    \begin{algorithmic}[1]
    \STATEx Input: $\vmweight, k,p, \schedule_{t-1}$
    \STATEx Output: $\schedule_t$
    \iftoggle{tech_report}
    {
        \STATE Do initialization /* see Algorithm \ref{alg:init} */
        \STATE Do path augmenting /* see Algorithm \ref{alg:path-aug} */
        \STATE Do checking a cycle /* see Algorithm \ref{alg:cycle} */
    }
    {
        \STATE \GREEN{Do initialization}
        \STATE \GREEN{Do path augmenting}
        \STATE \GREEN{Do checking a cycle}
    }
    \STATE Do back-propagating and scheduling 
    \end{algorithmic}
    \end{algorithm}
}

\iftoggle{tech_report}
{
    In the following, we explain the augmentation algorithm.
}
{
    We overview the augmentation algorithm. For the detailed description, we refer to~\cite{bui_ton09} or our technical report~\cite{tr}.
}
We start with some definitions. Given a matching $S$, an \textit{augmentation} $A$ of matching $S$ is a path or cycle where every alternate link is in $S$ and has the property that if all links in $A \cap S$ are removed from $S$ and all links in $A-S$ are added to that $S$, then the resulting set of links is another matching in $G$. The latter process of finding new matching is called \textit{augmenting} $S$ with $A$, and the resulting matching is denoted by $S\oplus A=(S -A) \cup (A-S)$.
A pair of augmentations $A_1$ and $A_2$ of matching $S$ is  $disjoint$ if no two links in $A_1 -S$ and $A_2 -S$ are adjacent, i.e., if they do not share a common node. Let $\mathcal{A}$ denote a set of disjoint augmentations of matching $S$ where every pair in $\mathcal{A}$ is disjoint. Then $\bigcup_{A \in \mathcal{A}}(S \oplus A)$ is also a matching in $G$. 

The overall procedure of the augmentation algorithm at each time slot $t$ is as follows. 
\begin{enumerate} 
    \item At the beginning of the time slot, each link $i$ is associated with some known weight $\mweight_i(t)$. 
    \item Given a valid matching $\schedule_{t-1}$ that is the schedule at time $t-1$, it randomly generates a set~$\mathcal{A}$ of disjoint augmentations of $\schedule_{t-1}$. 
    \item It compares the weight sum of $A-\schedule_{t-1}$ and $A \cap \schedule_{t-1}$ for each $A \in \mathcal{A}$. Let $B(A)$ be the one with the larger weight sum among the two. 
    \item It takes the new schedule $\schedule_t$ as $\bigcup_{A \in \mathcal{A}} B(A)$.
\end{enumerate}
For the comparison of the weight sum in the 3rd step, we define the gain of augmentation $A$ as
\begin{equation} \label{eq:gain}
\textstyle 
    G_t(A)
    %=\sum_{i \in A-\schedule_{t-1}} q_{i}(t)\cdot \mu_i - \sum_{i \in A \cap \schedule_{t-1}}q_{i}(t) \cdot \mu_i,
    =\sum_{i \in A-\schedule_{t-1}} w_i - \sum_{j \in A \cap \schedule_{t-1}} w_j,
\end{equation}
and obtain new schedule $\schedule_t$ by augmenting $\schedule_{t-1}$ with all $A \in \mathcal{A}$ of $G_t(A) > 0$.

%%%%%%%%%%%%%%%%%%%%%%%%%%%%%%%%%%%%%%%
\iftoggle{tech_report}
{
    %%%%%%%% by 190423 Daehyun Park
    % %%%%%%%% by 190416 Daehyun Park
    %%%%%%%% new try for Augmenting algorithm
    \begin{algorithm}[t]
    \caption{Initialization.}
    \begin{algorithmic}[1]
    \STATE $\tau \leftarrow 1$   ~~/* variable that counts mini-slots */
    \STATE $state_v \leftarrow \textsc{null}$
    \STATE $seed_v \leftarrow 1$ with probability $p$
    \STATE $alt_v \leftarrow \textsc{false}$ ~~/* for alternating links */
    \STATE $A_v \leftarrow \emptyset$  ~~/* list of links in augmentation */
    \IF{$seed_v=1$}
        \STATE $state_v \leftarrow \textsc{active}$ 
        \STATE $Z_v \leftarrow 0$ ~~/* size of augmentation */
        \STATE $G_v\leftarrow 0$ ~~/* gain of augmentation */
        \STATE Choose $\bar{Z}_v \in [k]$ at random ~~/* max aug. size */
        \IF{$\exists$ $n\in \mathcal{N}(v)$ s.t. $(v,n)\in \schedule_{t-1}$}
            \STATE Append link $(v,n)$ to $A_v$
            \STATE $G_v\leftarrow G_v -\mweight_{(v,n)}$ %\BLUE{\sout{x_{rand}\leftarrow \neg x_{rand}} }
            \STATE Send REQ$_v$ $\langle A_v,G_v, Z_v, \bar{Z}_v, \neg alt_v \rangle$ to $n$
        \ELSE
            \STATE Choose $n\in\mathcal{N}(v)$ at random
            \STATE Send REQ$_v$ $\langle A_v,G_v, Z_v, \bar{Z}_v, alt_v \rangle$ to $n$
        \ENDIF
        \STATE $state_v \leftarrow \textsc{wait}$ ~ /* wait ACK$_n$ from $n$ */
    %\ENDIF
    \ELSE ~~ /* $state_v = \textsc{null}$ */
        \IF{successfully receive REQ$_u$} 
            \STATE Set $A_v,G_v,Z_v,\bar{Z}_v,alt_v$ as in REQ$_u \langle \dots \rangle$
            %\STATE $state_v \leftarrow active$
            \STATE $state_v \leftarrow \textsc{active}$
            \STATE Send ACK$_v$ to $u$
        \ENDIF
    \ENDIF
    \IF{$state_v = \textsc{wait}$} 
        \IF{receive ACK$_n$}
            \STATE $state_v \leftarrow \textsc{used}$ /* inactive */
        \ELSE ~~ /* REQ collision occurred at $n$ */
            \STATE $state_v \leftarrow \textsc{done}$ /* terminus */
        \ENDIF
    \ENDIF
    \end{algorithmic}
    \label{alg:init}
    \end{algorithm}
}
%%%%%%%%%%%%%%%%%%%%%%%%%%%%%%%%%%%%%%%

%%%%%%%%%%%%%%%%%%%%%%%%%%%%%%%%%%%%%%%
\iftoggle{tech_report}
{
    \begin{algorithm}[t]
    \caption{Path augmenting.}
    \begin{algorithmic}[1]
    \FOR{$\tau \leftarrow 2$ to $2k+1$}
        \IF{$state_v = \textsc{active}$}
            \IF{$alt_v = \textsc{true}$}
                \STATEx ~~~~~~~~ /* if the last added link is in $\schedule_{t-1}$ */
                \IF{$\exists n\in \mathcal{N}(v)-A_v$ and $Z_v < \bar{Z}_v$}
                    \STATE Choose $n \in \mathcal{N}(v)-A_v$ at random
                    %\STATE $x_{rand}\leftarrow \neg x_{rand}$  
                    \STATE Send REQ$_v$ $\langle A_v,G_v, Z_v, \bar{Z}_v, \neg alt_v \rangle$ to $n$
                    \STATE $state_v \leftarrow \textsc{wait}$
                \ELSE
                    \STATE $state_v \leftarrow \textsc{done}$
                \ENDIF
            \newline \vspace{-0.4cm}
            \ELSE  ~~ /* if the last added link is not in $\schedule_{t-1}$ */
                \STATE $u \leftarrow$ the previous active node
                \STATE $G_v \leftarrow G_v + \mweight_{(u,v)}$
                \STATE Append link $(u,v)$ to $\mathcal{A}_v$
                \STATE $Z_v \leftarrow Z_v + 1$ 
                \IF {$\exists$ $n\in \mathcal{N}(v)$ s.t. $(v,n)\in \schedule_{t-1}$}
                    \STATE Append link $(v,n)$ to $A_v$  
                    \STATE $G_v\leftarrow G_v-\mweight_{(v,n)}$ %,~x_{rand}\leftarrow \neg x_{rand}
                    \STATE Send REQ$_v$ $\langle A_v, G_v, Z_v, \bar{Z}_v, \neg alt_v \rangle$ to $n$
                    \STATE $state_v \leftarrow \textsc{wait}$
                \ELSE
                    \STATE $state_v \leftarrow \textsc{done}$
                \ENDIF
            \ENDIF
        \ELSIF{$state_v = \textsc{null}$}
            \IF{successfully receive REQ$_u$} 
                \STATE Set $A_v,G_v,Z_v,\bar{Z}_v,alt_v$ as in REQ$_u \langle \dots \rangle$
                %\STATE $state_v \leftarrow active$
                \STATE $state_v \leftarrow \textsc{active}$
                \STATE Send ACK$_v$ to $u$
            \ENDIF
        \ENDIF
        \IF{$state_v = \textsc{wait}$} 
            \IF{receive ACK$_n$}
                \STATE $state_v \leftarrow \textsc{used}$ /* inactive */
            \ELSE ~~ /* REQ collision occurred at $n$ */
                \STATE $state_v \leftarrow \textsc{done}$ /* terminus */
            \ENDIF
        \ENDIF
    \ENDFOR
    \IF{$state_v = \textsc{active}$}
        \STATE $state_v \leftarrow \textsc{done}$
    \ENDIF
    \end{algorithmic}
    \label{alg:path-aug}
    \end{algorithm}
}
%%%%%%%%%%%%%%%%%%%%%%%%%%%%%%%%%%%%%%%

%%%%%%%%%%%%%%%%%%%%%%%%%%%%%%%%%%%%%%%
\iftoggle{tech_report}
{
    \begin{algorithm}[t]
    \caption{Check a cycle.} \label{alg:cycle}
    \begin{algorithmic}[1]
    \STATE $S^* \leftarrow \emptyset$
    \STATE $A=\mathbb{I}\{\text{the first node $n$ of $A_v\in$  $\mathcal{N}(v)$} ~\text{and}~ v \in A_v \}$
    \STATE $B=\mathbb{I}\{\text{Both end links of $A_v$ $\in$ $\schedule_{t-1}$} ~\text{and}~ v \in A_v\}$
    \STATE $C=\mathbb{I}\{Z_v<\bar{Z}_v ~\text{and}~ v \in A_v\}$
    \IF{$state_v=\textsc{done}$ and $A \cdot B \cdot C=1$}
        \STATE Append link $(v,n)$ to $A_v$
        \STATE $G_v\leftarrow G_v+\mweight_{(v,n)}$ 
    \ENDIF
    \IF{$G_v>0$}
        \STATE $S^* \leftarrow (A_v - S_{t-1}) \cup (S_{t-1}-A_v)$
    \ENDIF
    % \IF{$G_v > 0$}
    %     \STATE \BLUE{$S^* \leftarrow S^* \cup (A_v \oplus S_{t-1})$}
    \end{algorithmic}
    \end{algorithm}
}
%%%%%%%%%%%%%%%%%%%%%%%%%%%%%%%%%%%%%%%

%%%%%%%%%%%%%%%%%%%%%%%%%%%%%%%%%%%%%
\begin{comment}
%%%%%%%%%%%%%%%%%%%%%%%%%%%%%%%%%%%%%
%%%%%%%% by 190416 Daehyun Park
%%%%%%%% new try for Augmenting algorithm
\begin{algorithm}[t]
\caption{Back-propagating and scheduling}
\begin{algorithmic}[1]
\STATE /*$end_n(A)$ is the $n$-th element from the end of list $A$ */
\STATE $\tau\leftarrow 2k+2$
\IF{$state_v=\textsc{done}$, $G_v>0$, and length of $A_v> 1$}
    \IF{$(end_2(A_v),end_1(A_v))\in \schedule_{t-1}$}
        \STATE Switch off 
        $(end_2(A_v),end_1(A_v))$
    \ELSE
        \STATE Switch on $(end_2(A_v),end_1(A_v))$
    \ENDIF
    \STATE Delete $end_1(A_v)$ from $A_v$
    \STATE Send REQ$_v$ $\langle A_v\rangle$ to $end_2(A_v)$
\ENDIF
\FOR{$\tau=2k+3$ to $4k+2$}
\IF{listen REQ and length of $A_v> 1$ }
     \IF{$(end_2(A_v),end_1(A_v))\in \schedule_{t-1}$}
        \STATE Switch off 
        $(end_2(A_v),end_1(A_v))$
    \ELSE
        \STATE Switch on $(end_2(A_v),end_1(A_v))$
    \ENDIF
    \STATE Delete $end_1(A_v)$ from $A_v$
    \STATE Send REQ$_v$ $\langle A_v\rangle$ to $end_2(A_v)$
\ENDIF
\ENDFOR
\end{algorithmic}
\label{alg:switching}
\end{algorithm}
%%%%%%%%%%%%%%%%%%%%%%%%%%%%%%%%%%%%%
\end{comment}
%%%%%%%%%%%%%%%%%%%%%%%%%%%%%%%%%%%%%

\begin{figure}[t]
\centering
    \includegraphics[width=3.5in]{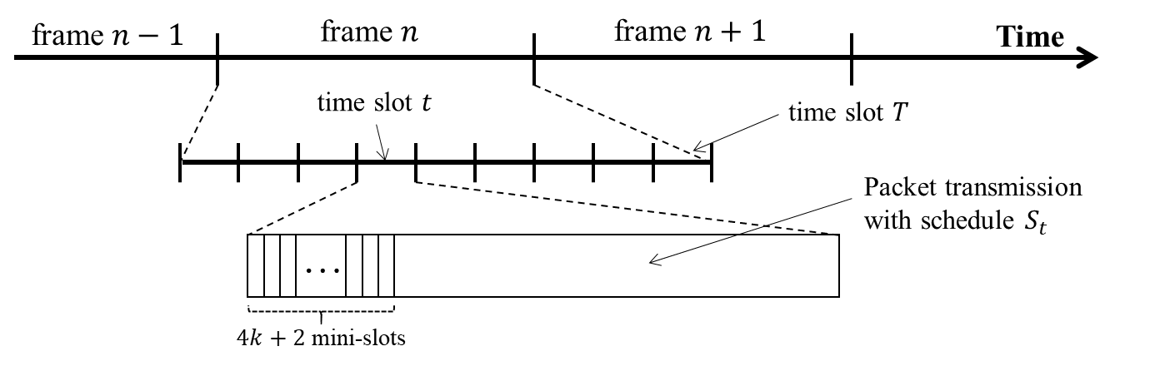}
    \caption{Time structure.}
\label{fig:time_structure}
\end{figure}

\iftoggle{tech_report}
{
    The augmentation algorithm can accomplish the above procedure in a distributed fashion. The algorithm needs two configuration parameters $p$ and $k$, and consists of the following four stages in each time slot as shown in Algorithm~\ref{alg:scheme}: initialization, augmenting, checking a cycle, and back-propagating/scheduling. For the ease of exposition, we consider additional time structure of mini-slots, as shown in Fig.~\ref{fig:time_structure}, and all the four stages end in $4k+2$ mini-slots.
}
{
    The augmentation algorithm can accomplish the above procedure in a distributed fashion. The algorithm has two configuration parameters $p$ and $k$, and consists of the following four stages in each time slot: initialization, augmenting, checking a cycle, and back-propagating/scheduling. For the ease of exposition, we consider additional time structure of mini-slots, as shown in Fig.~\ref{fig:time_structure}, and all the four stages end in $4k+2$ mini-slots.
}
\begin{enumerate}[i)]
    \iftoggle{tech_report}
    {
        \item In the initialization stage (Algorithm~\ref{alg:init}), at mini-slot $\tau=1$, each node~$v$ selects itself as a seed with probability~$p$ (line 3). Once selected, it becomes an \emph{active} node. It starts an augmentation $A_v$ and randomly selects $\bar{Z}_v \in [1,k]$, which is the maximum size of $A_v$. (lines 6-10), and adds the first link of the augmentation from its neighbors $\mathcal{N}(v)$ (lines 11-18). The rest part is the coordination with immediate neighbors by exchanging a request (REQ) and an acknowledgement (ACK) to join in the newly created augmentation, and by informing the added node of necessary information such that it continues building the augmentation (lines 19-28). Note that if a seed node has a link in $\schedule_{t-1}$, then the link should be the first link in the augmentation (line 11). Otherwise, a link is randomly chosen (line 16).
        %the augmentation task will be undertaken by the end node of the added link (lines 18-25). 
        At the end of the mini-slot, node $v$ becomes inactive and node $n$ becomes active.
        \item The second augmenting stage (Algorithm \ref{alg:path-aug}) extends the augmentation by adding a link at each mini-slot, which is either a random neighboring link outside $\schedule_{t-1}$ (lines 3-9) or a neighboring link in $\schedule_{t-1}$ (lines 11-21). Note that, for a link that is not in $\schedule_{t-1}$, gain $G_v$ and list $A_v$ are updated after $REQ$ message is received, and for a link in $\schedule_{t-1}$, they are updated before REQ message is sent (line 16). The extension continues until the augmentation reaches size $\bar{Z}_v$ (line 4) or the augmentation cannot be extended any more. It also needs to pay attention that we count the size of augmentation $A_v$ as the number of new links, i.e., $Z_v = |A_v - \schedule_{t-1}|$ (line 14). The rest part is similar to those of the initialization stage (lines 22-31).
        \item Once the augmenting stage finishes, it checks whether the resulting augmentation makes a cycle or not (Algorithm~\ref{alg:cycle}). For a terminus node (i.e., node $v$ with $state_v = \textsc{done}$), it can be easily verified that the augmentation is a cycle if $A=B=C=1$ in lines~2-4. In the case of a cycle augmentation, the gain is updated to include the last link (line 6), and if the final gain is positive, the link is included in new schedule (line~8).
    }
    {
        \item {\bf Initialization stage:} At mini-slot $\tau=1$, each node~$v$ selects itself as a seed with  probability~$p$. Once selected, it becomes an \emph{active} node. It starts an augmentation $A_v$ and randomly selects $\bar{Z}_v \in [1,k]$, which is the maximum size\footnote{The size $Z$ of an augmentation $A$ is defined as the number of new links in $A$, i.e., $Z = |A - \schedule_{t-1}|$.} of $A_v$. Then, it adds the first link to $A_v$ from its neighboring links $\mathcal{N}(v)$: if there is a link in $\schedule_{t-1} \cap \mathcal{N}(v)$, then we include the link in $A_v$, and otherwise, we randomly choose a link from $\mathcal{N}(v)$. Once the first link $(v,n)$ is chosen, the two nodes $v$ and $n$ coordinate with each other by exchanging necessary information through request (REQ) and acknowledgement (ACK) messages. The information (including current $A_v$, current $G_t(A_v)$ according to (\ref{eq:gain}), $\bar{Z}_v$, etc) allows node $n$ to continue building the augmentation. At the end of the first mini-slot, node $v$ becomes inactive and node $n$ becomes active.
        \item {\bf Augmenting stage:} It extends $A_v$ by adding a link at each mini-slot $\tau \in [2, 2k+1]$. Current active node selects new link that should be either a random neighboring link outside $\schedule_{t-1}$ (if the previously selected link is in $\schedule_{t-1}$) or a neighboring link in $\schedule_{t-1}$ (if the previously selected link is not in $\schedule_{t-1}$). It updates $G_t(A_v)$ and $A_v$ accordingly, and proceeds a similar coordination through REQ and ACK messages, and at the end of mini-slot, active node changes to the corresponding end node of the new link. The extension continues until one of the following conditions hold: the size of $A_v$ equals $\bar{Z}_v$, $A_v$ cannot be extended any more, or the extension fails due to message collision (see Fig.~\ref{fig:augmenting}).
        \item {\bf Cycle-checking stage:} When the augmenting stage finishes, the last node is called the \emph{terminus}. The terminus checks whether the final augmentation $A_v$ forms a cycle or not. %If node $n$ is the last node when the augmenting stage finishes (called a terminus), it can verify whether the augmentation consists of a cycle or not. It is a cycle, if all the following three hold: (i) node $n$ of $A_s$ belongs to $\mathcal{N}(s)$, where $s$ is the seed of that augmentation, (ii) both end links of $A_s$ belong to $\schedule_{t-1}$, and (iii) the current size $Z_v$ is strictly less than $\bar{Z}_v$. 
        If it forms a cycle, the gain is updated to include the last link $(n,s)$.
    }
    \item {\bf Back-propagating/scheduling stage:} The last stage is for back-propagating the final gain from the terminus to the seed through $A_v$, and in the meantime, constructing $S^*$ either by scheduling links outside $\schedule_{t-1}$ if the gain is positive, or by scheduling links in $\schedule_{t-1}$ if the gain is negative. This takes additional $2k+1$ mini-slots at most. The final result $S^*$ will be used as the new schedule $\schedule_t$ during time slot $t$.
\end{enumerate}

\noindent \emph{Remarks:} In our description, we present $\schedule_{t-1}$ as if it is a global variable, but each node $v$ indeed requires only the local view of it, i.e., $\mathcal{N}(v) \cap \schedule_{t-1}$, which can be obtained during the back-propagation in the last stage of the previous time slot. 
After all the stages, a set of augmentations (one per a seed node) will be generated. Since a size-$\bar{Z}$ augmentation $A$ can have at most $\bar{Z}+1$ links of $\schedule_{t-1}$ and $\bar{Z}$ new links, the number of total links in $A$ can be up to $2\bar{Z}+1 \le 2k + 1$.

\begin{figure}[t]
\centering
\includegraphics[width=3.3in]{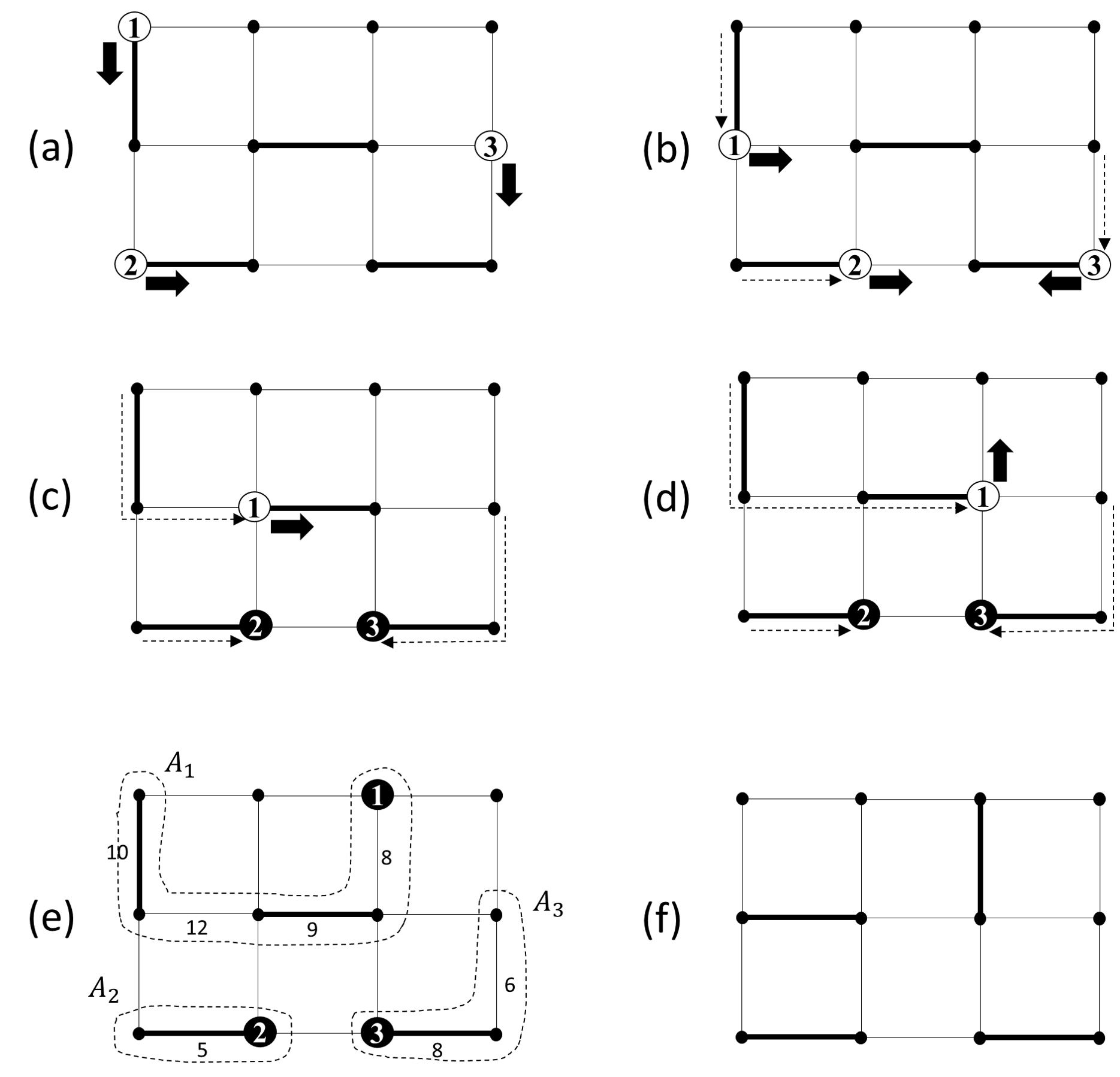}
\caption{
%\BLUE{Example operation of Augmentation algorithm when $k=2$. Circled numbers are active nodes at each step. In first four graphs, Bold lines are links in $S_{t-1}$ and short bold arrows indicate which direction the active nodes want to go extend. Chains of dashed arrows are augmentations being extending or not. In fifth graph, chains of long bold arrows are augmentations which finished extending. In the last graph, bold lines are links in $S_t$.}
    Example operation of the augmentation algorithm with $k=2$ in a time slot. After the cycle-checking stage, link weights of final augmentations are shown in (e). After back-propagation, links are accordingly augmented with $A_1$, and the final schedule is marked by thick solid lines in (f).}   
% \RED{( (i) plz add a number on subfigs (or add big arrows between subfigs) such that one can see the order of subfigs. (ii) make thick link thicker. (iii) in 5-th subfig, how about adding a wrappting dotted line for each 'completed' augmentation.)}}
\label{fig:augmenting}
\end{figure}

Fig.~\ref{fig:augmenting} illustrates an example operation of the augmentation algorithm during time slot $t$ in a $3\times4$ grid topology. Nodes are dots, and solid lines are links. The previous schedule $\schedule_{t-1}$ is marked by thick solid lines in Fig.~\ref{fig:augmenting}(a)-(e). At the beginning ($\tau=1$), each node selects itself as a seed with probability $p$. In this example, three nodes are selected and become an active as marked by (white) numbered circles in Fig.~\ref{fig:augmenting}$(a)$. Active nodes $1$ and $2$ have to start its augmentation with the previous scheduled link, while active node $3$ selects one of three neighboring nodes at random. The solid arrow denotes the link selected by the active node. Once selected, the nodes exchange the necessary information. 
In the next mini-slot ($\tau=2$), active nodes change as shown in Fig.~\ref{fig:augmenting}$(b)$. Narrow dotted arrows denote the augmentation up to now. The active nodes continue to build the augmentation repeating the procedure, until the augmentation cannot be extended or it reaches the maximum size.
%In the second mini-slot, current active nodes of $1$ and $2$ will randomly select one of its neighboring node as next active nodes, and current active node of $3$ follows the previously scheduled link and selects the next active node in a deterministic manner, as shown by thick short arrows. Then they exchange REQ/ACK and active nodes change in the next mini-slot. 
%The procedure repeats until $2k+1$-th mini-slot, or the augmentation cannot be extended.
%\footnote{If active node $3$ in Fig.~\ref{fig:augmenting}(a) selects the upward node as next active node, then new active node cannot find any link in $\schedule_{t-1}$, and terminates the augmentation.}. 
%
In the meantime, if two augmentations collide as shown by active nodes $2$ and $3$ in Fig.~\ref{fig:augmenting}$(b)$, both augmentations terminate. The node at the collision point will belong to the augmentation that follows a link in $\schedule_{t-1}$, as shown in Fig.~\ref{fig:augmenting}(c). The terminus nodes are marked by solid (black) number circles.
The result after $2k+1$-th mini-slots is shown in Fig.~\ref{fig:augmenting}(e), where each of three augmentations is marked by a dotted enclosure. The number of links in the augmentations denote their weight. 
Then after checking a cycle, the back-propagating stage follows. Each terminus makes the final decision by comparing the weight sum as in~\eqref{eq:gain}. The decision propagates backward through the augmentation, and leads to new schedule $\schedule_{t}$ as shown in Fig.~\ref{fig:augmenting}(f).

By using the aforementioned augmentation algorithm in matching selection of line 8 in Algorithm~\ref{alg:generic_ucb}, we complete A$^k$-UCB scheme. However, evaluating the performance of A$^k$-UCB is not straightforward. Since the scheduler do not know the true link rate $\mu_i$ and instead use the experience-based UCB index $\bar{w}_{i,t}$, the previous analysis of~\cite{bui_ton09} for scheduling performance is not applicable due to inaccurate link rate information. Also for learning performance, considering A$^k$-UCB as an $(\alpha,\beta)$-oracle is not much helpful due to the fact that probability $\beta$ can be arbitrarily small as the network scales up.
%On the other hand, the augmentation algorithm can serve as an $(\alpha, \beta)$-approximation oracle that takes weight $\vmweight$ as input, and outputs matching $S$ such that $P\{ S \in \nopts_{\vmweight} \} > \beta$ with $\alpha = \alpk$. By adopting this oracle, one can develop a learning algorithm that achieves the logarithmic growth of $\alpha\beta$-regret with $\alpha = \alpk$~\cite{chen_icml2013}. However, as the network scales, the probability~$\beta$ of the augmentation algorithm can be very small (see Lemma~\ref{lem:probability}), leading to a performance bound that is not much meaningful. 

Our main contribution is to analytically characterize the performance of our joint learning and scheduling scheme A$^k$-UCB, and to show that it can achieves the rate-optimal logarithmic growth of $\alpk$-regret regardless of the network size in the learning, and further it has the close-to-optimal stability region that equals $\alpk \Lambda$ in the scheduling.

\section{Performance Evaluation} \label{sec:evaluation}

We first consider the regret performance of A$^k$-UCB in a single frame, and then evaluate its scheduling performance across frames.

\subsection{Regret performance in a single frame}

We show that A$^k$-UCB has distribution-dependent upper bound of $O(\log \framel)$ on the regret in a frame of length~$\framel$. 
\iftoggle{tech_report}
{
    We start with the following lemmas. 
    \begin{lemma} \label{lem:possibility}
        Given matching $\schedule_{t-1}$, weight $\viweight_t$, and a fixed $k > 0$,
        the augmentation algorithm can generate a set $\augstar$ of disjoint augmentations of $\schedule_{t-1}$ such that 
        \begin{equation*} \textstyle
            %\totalmw_{\vmweight}( \schedule_{t-1}\oplus\augstar) \geq \left( \alpk \right) \cdot \maxtotalmw_{\vmweight},
            \totalmw_{\viweight_t}( \schedule_{t-1}\oplus\augstar) \geq \left( \alpk \right) \cdot \maxtotalmw_{\viweight_t},
        \end{equation*}
        and every augmentation $A$ in $\augstar$ has a size no greater than $k$. 
    \end{lemma}
    %Lemma \ref{lem:possibility} also holds with any weight, including $\vmweight$.
}
{
    We start with the following lemma.
    %\begin{lemma} \label{lem:possibility}
    %    Given matching $\schedule_{t-1}$, weight $\viweight_t$, and a fixed $k > 0$,    the augmentation algorithm can generate a set $\augstar$ of disjoint    augmentations of $\schedule_{t-1}$ such that 
    %    \begin{equation*} \textstyle
            %\totalmw_{\vmweight}( \schedule_{t-1}\oplus\augstar) \geq \left( \alpk \right) \cdot \maxtotalmw_{\vmweight},
    %        \totalmw_{\viweight_t}( \schedule_{t-1}\oplus\augstar) \geq \left( \alpk \right) \cdot \maxtotalmw_{\viweight_t},
    %    \end{equation*}
    %    and every augmentation $A$ in $\augstar$ has a size no greater than $k$.
    %\end{lemma} 
    %Lemma~\ref{lem:possibility} also holds when using $\vmweight$ instead to $\viweight_t$.
%    We start with the following proposition, which states that the augmentation algorithm is an $(\alpha,\beta)$-approximation oracle with $\alpha = \alpk$ and $\beta = \delta$ for some $\delta > 0$. The proof follows the same line of analysis of~\cite{bui_ton09}, and can be found in~\cite{tr}.
}
\begin{lemma}\label{lem:probability} 
Given any $\schedule_{t-1}$, weight $\viweight_t$, and a fixed $k > 0$, there exists $\delta > 0$ such that, with probability at least $\delta$, the augmentation algorithm generates a set $\augstar$ of disjoint augmentations that satisfies $(\schedule_{t-1} \oplus \mathcal{A}^*) \in \nopts_{\viweight_t}$, i.e., $\Pr \{ (\schedule_{t-1} \oplus \mathcal{A}^*) \in \nopts_{\viweight_t}\} \ge \delta$, or equivalently,
    \begin{equation} \label{eq:alpha-delta} \textstyle 
        \Pr \left\{ 
            %\totalmw_{\vmweight}(\schedule_{t-1} \oplus \augstar) \geq \left( \alpk \right) \cdot \maxtotalmw_{\vmweight}
            \totalmw_{\viweight_t}(\schedule_{t-1} \oplus \augstar) \geq \alpk \cdot \maxtotalmw_{\viweight_t}
            \right\} \geq \delta,
    \end{equation}
    where $\delta \ge \min \{ 1, ( \frac{p}{1-p} )^{|\setNode|} \} \cdot ( \frac{1-p}{k\Sigma})^{|\setNode|}$, $|\setNode|$ is the number of nodes, and $\Sigma$ is the maximum node degree.
\end{lemma}
\noindent Lemma~\ref{lem:probability} means that the augmentation algorithm is an $(\alpha,\beta)$-approximation oracle with $\alpha = \alpk$ and $\beta > 0$, where $\beta$ can be arbitrarily small according to the network size.
\iftoggle{tech_report}
{
    The proofs of the lemmas follow the same line of analysis of~\cite{bui_ton09} and can be found in Appendices~\ref{app:A} and~\ref{app:B}, respectively.
}
{
    The proof follows the same line of analysis of~\cite{bui_ton09} and thus omitted. For the completeness, we provide the proof in~\cite{tr}.
}

Next, we need a generalized version of the decomposition inequality for $\alpha$-regret. From (\ref{eq:a-regret}), we have
\begin{equation} \label{eq:decomposition}
\begin{split}
    \alpR(t) 
    &\textstyle = t \cdot \alpha \cdot \maxtotalmw_{\vmweight} 
        - \mathbb{E}\left[ \sum_{\tau=1}^t \totalmw_{\vmweight}(\schedule_{\tau}) \right]\\
    %&= \sum_{{\tau}=1}^t \left[ \alpha\cdot \maxtotalmw_{\vmweight}
    %    - \mathbb{E} [ \totalmw_{\vmweight}(\schedule_{\tau}) ] \right]\\
    &\hspace{-0.5cm}\textstyle = \sum_{{\tau}=1}^t \sum_{S \in \setS}
        \mathbb{E}\left[ \mathbb{I} \{ \schedule_{\tau}=S \} \cdot ( \alpha \maxtotalmw_{\vmweight} - \totalmw_{\vmweight}(S) ) \right]\\
    %&= \sum_{S \in \setS} \mathbb{E} 
    %    \left[ \sum_{{\tau}=1}^t\mathbb{I}(\schedule_{\tau}=S) \right] \cdot ( \alpha\cdot \maxtotalmw_{\vmweight} - \totalmw_{\vmweight}(S) )\\
    &\hspace{-0.5cm}\textstyle \leq \sum_{S \in \setS} \mathbb{E} [ \nochosen_S(t) ]
        \cdot\alpgapmax,
\end{split}
\end{equation}
where $\alpgapmax=\alpha\cdot \maxtotalmw_{\vmweight}-\min_{S\in\cnopts} \totalmw_{\vmweight}(S)$ is the maximum near-optimal gap. Similarly we define the minimum near-optimal gap $\alpgapmin=\alpha\cdot \maxtotalmw_{\vmweight}-\max_{S\in\cnopts} \totalmw_{\vmweight}(S)$.

%Since we use the augmentation algorithm with UCB index $\viweight_t$ instead of true weight $\vmweight$, which is unknown a priori, we need a new technical approach different from~\cite{bui_ton09}. Recall that 
The following lemma ensures that, if a non-near-optimal matching in $\cnopts_{\vmweight}$ is played many times, then its index sum is smaller than that of any near-optimal matching in $\nopts_{\vmweight}$. 
\begin{lemma} \label{lem:nearoptimal}
Given a frame of length $\framel$, if a non-near-optimal matching $S \in \cnopts_{\vmweight}$ is played more than $l_T=\lceil \frac{ 4\noLink^2(\noLink+1) \ln T }{ \alpgapmin } \rceil$ times by $t$-th time slot in the frame, then the probability that the total sum of UCB indices over $S$ at time slot $t$ is greater than that over any near-optimal matching $S' \in \nopts_{\vmweight}$ is bounded by
\begin{equation}
\Pr \left\{ \totalmw_{\viweight_t}(S)  \geq \totalmw_{\viweight_t}(S') \right\} 
    \leq 2\noLink t^{-2},
\end{equation}
for all $t \leq \framel$ such that $\nochosen_S(t)\geq l_\framel$.
\end{lemma}
\noindent We emphasize that $\nopts_{\vmweight}$ and $\cnopts_{\vmweight}$ are defined with true weight $\vmweight$, while the matching comparison is based on UCB index $\viweight_t$. The lemma shows that the augmentation algorithm may still work well, even when the true weight is replaced with the UCB index. The proof of the lemma is analogous to Lemma $A.1$ of~\cite{kang_tmc20} but has some differences due to 'near-optimality'.
\iftoggle{tech_report}
{
    It can be found in Appendix~\ref{app:C}.
}
{
    It can be found in~\cite{tr}.
}

One of our main results, the regret bound of A$^k$-UCB, can be obtained as follows.
\begin{proposition} \label{pro:regretbound}
For a network graph $\mathcal{G}=(\setNode,\setLink)$, A$^k$-UCB achieves the regret performance bound of
\begin{equation*}
\begin{split}
    \alpR(t) 
        &\textstyle \leq \alpgapmax \left[ D_1 
        \cdot \frac{\log t}{(\alpgapmin)^2} + D_2
        %\frac{4\noLink^2(\noLink+1)\ln t}{(\alpgapmin)^2} \cdot \left( \vert \setS \vert-1 \right)
        \right],
        %&~~~+ \alpgapmax \cdot (D_1 + D_2),
\end{split}
\end{equation*}
for all $t \in \{1,...,\framel\}$, where 
%$p$ and $k$ denote the seed probability and the maximum augmentation length of the augmentation algorithm, respectively. Further, 
$D_1 = \left( 1 + \frac{1}{\delta} \right) \cdot 4\noLink^2(\noLink+1) \cdot \left( \vert \setS \vert-1 \right)$, 
$D_2 = \frac{\vert\setS\vert-1}{\delta} (1+ \frac{\noLink\delta\pi^2}{3} + \frac{\noLink(\vert\setS\vert-2)\pi^2}{6} ) + \frac{1-\delta}{\delta} + \frac{2\noLink\pi^2}{3\delta}$,
$\alpha = \alpk$, and $\delta = \min \{ 1,( \frac{p}{1-p} )^{|\setNode|} \} \cdot ( \frac{1-p}{k\Sigma} )^{|\setNode|}$.
\end{proposition} 
\noindent It shows that A$^k$-UCB achieves the logarithmic growth $O(\log \framel)$ of $\alpha$-regret.
Note that although the result is somewhat similar to those in~\cite{chen_icml2013, gai_ton2012}, their proof techniques are not applicable. Suppose that at time slot $t$, A$^k$-UCB randomly generates a set $\mathcal{A}_t$ of augmentations based on the previous schedule $\schedule_{t-1}$, and $\mathcal{A}_t$ consists of a single augmentation for the ease of exposition. It is possible that both the matchings, previous schedule $\schedule_{t-1}$ and schedule $\schedule_{t-1} \oplus \mathcal{A}_t$ generated by the augmentation algorithm, are non-near-optimal. This implies that we cannot ensure that the index sum of the chosen schedule (i.e., either $\schedule_{t-1}$ or $\schedule_{t-1} \oplus \mathcal{A}_t$) is greater than $\alpha \maxtotalmw_{\viweight_t}$, because the index-sum comparison is done only between the two non-near-optimal matchings. This, the lack of comparison with the optimal matching (or near-optimal matching in our case) at every time slot, makes the previous regret analysis technique non-applicable. 
\iftoggle{tech_report}
{ 
    We successfully address the difficulties by grouping the plays of non-near-optimal matchings. The detailed proof can be found in Appendix~\ref{app:D}.
}
{
    We successfully address the difficulties by grouping the plays of non-near-optimal matchings.
    
%%%%%%%%%%%%%%%%%%%%%%%%%%%%
% app:D
%\section{Proof of Proposition~\ref{pro:regretbound}} \label{app:D}

Overall, we show that the number of explorations to non-near-optimal matchings is bounded. To this end, we consider a sequence of time points where a non-near-optimal matching is sufficiently played at each point. They serve as a foothold to count the total number of plays of non-near-optimal matchings.

To begin with, for an arbitrary fixed time $h>0$, let $l_h = \lceil \frac{4\noLink^2(\noLink+1)\ln h}{(\alpgapmin)^2} \rceil$, and let $\hat{t}_h$ denote the first time when all non-near-optimal matchings are sufficiently (i.e., more than $l_h$ times) explored, i.e., 
\begin{align*}
\hat{t}_h = \min \left\{ t ~\middle|~ 
    \nochosen_S(t) \geq  l_h ~\text{for all}~ S \in \cnopts_{\vmweight} \right\}. 
\end{align*}

\noindent \textbf{(1) When} $\hat{t}_h \leq h :$ 
Let $\cnopts_{\vmweight} = \{S^1,S^2,...,S^M\}$ with $M = \vert \cnopts_{\vmweight} \vert$. Further we define $\overline{\setS}(t) = \{ S \in \cnopts_{\vmweight} ~|~ \nochosen_S(t) \geq l_h\}$, which is the set of non-near-optimal matchings that are scheduled sufficiently many times by time $t$, and $\underline{\setS}(t) = \cnopts_{\vmweight} - \overline{\setS}(t)$ denotes the set of not-yet-sufficiently-scheduled non-near-optimal matchings.  
Also, let $t^n$ denote the time when matching $S^n$ is sufficiently scheduled, i.e., $\nochosen_{S^n}(t^n) = l_h$. Without loss of generality, we assume $t^1 < t^2 < \dots < t^M = \hat{t}_h$.

To apply the decomposition inequality (\ref{eq:decomposition}), we need to estimate the expected value of $\sum_{S \in \cnopts_{\vmweight}} \nochosen_S(\hat{t}_h)$, which can be written as
\begin{equation} \label{eq:tau_samplepath}
\begin{split}
&\textstyle \sum_{S \in \cnopts_{\vmweight}} \nochosen_S(\hat{t}_h) 
    \textstyle= \sum_{S \in \cnopts_{\vmweight}} \sum_{t=1}^{\hat{t}_h} \mathbb{I}\{\schedule_t=S\} \\
    %&= \sum_{j=1}^M \sum_{t=1}^{\hat{t}_h} \mathbb{I}\{\schedule_t=S^j\} \\
    &\textstyle = l_h M + \sum_{n=1}^{M-1}\sum_{t=t^n+1}^{t^{n+1}} \sum_{S \in \overline{\setS}(t^n)} \mathbb{I}\{ \schedule_t =S\}.
    %&= l_h M + \sum_{n=1}^{M-1} \sum_{t=t^n+1}^{t^{n+1}} \sum_{j=1}^n  \mathbb{I}\{ \schedule_t =S^j \}.
\end{split}
\end{equation}
%
%\iftoggle{tech_report}
%{
    Hence, we need to estimate $\sum_{S \in \overline{\setS}(t^n)} \Pr \{ \schedule_t =S \}$ for $t \in (t^n, t^{n+1}]$, which can be obtained as in the following lemma.
%}

\begin{lemma} \label{lem:prob-for-sufficiently-played-matching}
For each $t \in (t^n, t^{n+1}]$, we have 
    \begin{equation} \label{eq:upper-bound-for-sufficiently-played-matching}
    \begin{split}
        &\textstyle \sum_{S \in \overline{\setS}(t^n)} \Pr \{ \schedule_t =S \}\\
            &\textstyle \leq (1-\delta) \cdot \sum_{S \in \overline{\setS} (t^n)} \Pr \{\schedule_{t-1}=S\} \\
            &\textstyle + \Pr \{\schedule_{t-1} \in \underline{\setS}(t^n) \} +\left( \vert \overline{\setS}(t^n) \vert + \delta \right) \cdot 2\noLink t^{-2}.
    \end{split}
    \end{equation}
\end{lemma}
\begin{proof}
We first divide the case into three exclusive sub-cases based on the previous schedule $\schedule_{t-1}$: 
%From $\overline{\setS}(t^n) \cup \underline{\setS}(t^n) \cup \nopts_{\vmweight} = \setS$, we define three 
events $\mathbb{A}=\{ \schedule_{t-1} \in \nopts_{\vmweight} \}$, $\mathbb{B}=\{\schedule_{t-1} \in \underline{\setS}(t^n)\}$, and $\mathbb{C}=\{\schedule_{t-1} \in \overline{\setS}(t^n)\}$. Then we have 
\begin{align}
&\textstyle \sum_{S \in \overline{\setS}(t^n)} \Pr \{ \schedule_t =S \} \nonumber \\
%&&\sum_{j=1}^n \Pr \{ \schedule_t = S^j \} \nonumber \\
    &= \textstyle \sum_{S \in \overline{\setS}(t^n)} \Pr \{ \schedule_t =S ~\vert~ \mathbb{A} \} \cdot \Pr \{ \mathbb{A} \} \label{eq:first}\\
    %&&~~= \sum_{j=1}^n \Pr \{ \schedule_t = S^j \vert \schedule_{t-1} \in \nopts_{\vmweight} \} \cdot \Pr (\schedule_{t-1} \in \nopts_{\vmweight}) \label{eq:first}\\
    &~~\textstyle + \sum_{S \in \overline{\setS} (t^n)} \Pr \{ \schedule_t =S ~\vert~ \mathbb{B}\} \cdot \Pr \{\mathbb{B}\} \label{eq:second} \\
    %&&~~~~+ \sum_{j=1}^n \Pr \{ \schedule_t = S^j \vert \schedule_{t-1} \in \underline{\setS}(t^n)\} \cdot \Pr \{\schedule_{t-1} \in \underline{\setS}(t^n)\} \label{eq:second} \\
    &~~\textstyle + \sum_{S \in \overline{\setS} (t^n)} \Pr \{ \schedule_t =S ~\vert~ \mathbb{C}\} \cdot \Pr \{\mathbb{C}\}. \label{eq:third}
    %&&~~~~+ \sum_{j=1}^n \Pr \{ \schedule_t = S^j \vert \schedule_{t-1} \in \overline{\setS}(t^n)\} \cdot \Pr \{\schedule_{t-1} \in \overline{\setS}(t^n)\}. \label{eq:third}
\end{align}
Let $\mathcal{A}_t$ denote the set of augmentations chosen under our algorithm at time $t$. We can obtain a bound on (\ref{eq:first}) as
\begin{align}
    &\textstyle\sum_{S \in \overline{\setS}(t^n)} \Pr \{ \schedule_t =S ~\vert~ \mathbb{A} \} \cdot \Pr \{\mathbb{A} \} \nonumber\\
    %&~~\textstyle\leq \sum_{S \in \overline{\setS}(t^n)} \Pr \{\schedule_{t-1} \in \nopts_{\vmweight}\} \cdot \Pr \{ S = \schedule_{t-1} \oplus \mathcal{A}_t \} \nonumber \\
    %&~~\textstyle\leq \sum_{S \in \overline{\setS}(t^n)} \Pr \{ \mathbb{A} \} \cdot \Pr \{ S = (\schedule_{t-1} \oplus \mathcal{A}_t) \} \nonumber \\
        %&~~~~~~~~~~~~ \textstyle \cdot \Pr \{ \totalmw_{\viweight_t}(S)  \geq \totalmw_{\viweight_t}(\schedule_{t-1}) ~|~ \schedule_{t-1} \in \nopts_{\vmweight} \} \nonumber \\
        %&~~~~~~~~~~~~~~~~ \textstyle \cdot \Pr \{ \totalmw_{\viweight_t}(S)  \geq \totalmw_{\viweight_t}(\schedule_{t-1}) ~|~ \mathbb{A} \} \nonumber \\
    %&~~\textstyle \leq \sum_{S \in \overline{\setS}(t^n)} \Pr \{\totalmw_{\viweight_t}(S)  \geq \totalmw_{\viweight_t}(\schedule_{t-1})~|~\schedule_{t-1} \in \nopts_{\vmweight} \} \nonumber\\
    &~~\textstyle \leq \sum_{S \in \overline{\setS}(t^n)} \Pr \{\totalmw_{\viweight_t}(S)  \geq \totalmw_{\viweight_t}(\schedule_{t-1})~|~\mathbb{A}\} \cdot \Pr \{ \mathbb{A} \} \nonumber\\
    &~~\leq \vert \overline{\setS}(t^n) \vert \cdot 2 \noLink t^{-2}, \label{eq:4}
\end{align}
where the last inequality comes from Lemma \ref{lem:nearoptimal}. The result holds for all $t \in (t^n,t^{n+1}]$. 
For the second term (\ref{eq:second}), we have
\begin{align}
\textstyle \sum_{S \in \overline{\setS} (t^n)} 
    %& \Pr \{ \schedule_t =S \vert \schedule_{t-1} \in \underline{\setS}(t^n)\} \cdot \Pr \{\schedule_{t-1} \in \underline{\setS}(t^n)\}\nonumber \\ 
    \Pr \{ \schedule_t =S ~\vert~ \mathbb{B} \} \cdot \Pr \{\mathbb{B} \}
    %&\leq \Pr \{\schedule_{t-1} \in \underline{\setS}(t^n) \}. \label{eq:5}
    \leq \Pr \{\mathbb{B}  \}. \label{eq:5}
\end{align}
Finally, the third term (\ref{eq:third}) denotes the probability to transit from a sufficiently-played non-near-optimal matching to a sufficiently-played non-near-optimal matching, and thus we have
\begin{align*}
    &\textstyle \sum_{S \in \overline{\setS} (t^n)} \Pr \{ \schedule_t = S ~\vert~ \mathbb{C} \} \cdot 
    \Pr \{ \mathbb{C} \} \\
        &\textstyle = \sum_{S \in \overline{\setS} (t^n)} \Pr \{ \schedule_t \in \overline{\setS}(t^n) \vert \schedule_{t-1}=S \} \cdot \Pr \{\schedule_{t-1}=S\}.
\end{align*}
%Then, for $\alpha = \alpk$ and $S \in \overline{\setS}(t^n)$, 
Letting $S' = S \oplus \mathcal{A}_t$ and using Lemma~\ref{lem:probability}, the conditional probability can be derived as
\[
\begin{split}
&\Pr \{ \schedule_t \in \overline{\setS}(t^n) \vert \schedule_{t-1}=S \}\\
    %&\le \Pr \{ \schedule_t \in \overline{\setS}(t^n) \vert \schedule_{t-1}=S, S' \in \nopts_{\vmweight} \} \cdot \Pr \{ S' \in \nopts_{\vmweight}\} \\
    %\\&~~~~~~+ 
    %\Pr \{ \schedule_t \in \overline{\setS}(t^n) \vert \schedule_{t-1}=S, S \oplus \mathcal{A}_t \in \cnopts_{\vmweight} \} \cdot 
    %&\hspace{6cm} + \Pr \{S' \in \cnopts_{\vmweight}\}\\
    &\leq \Pr \{ \schedule_t \in \overline{\setS}(t^n) \vert \schedule_{t-1}=S, S' \in \nopts_{\vmweight} \} \cdot \delta + (1-\delta)\\
    &= \Pr \{\totalmw_{\viweight_t}(S) \geq \totalmw_{\viweight_t}(S')\}\cdot \delta + (1-\delta)\\
    %&~= \Pr \left\{\sum_{e \in S} \ucbI_{e,t} \geq \sum_{e \in C_t}\ucbI_{e,t}\middle|~S \in \overline{\setS}(t^n),~ C_t\in \nopts_{\vmweight}\right\}\cdot \delta +1-\delta\\
    &\leq 2\noLink t^{-2}\cdot \delta + (1-\delta). 
\end{split}
\]
where %the second inequality comes from Lemma~\ref{lem:probability}, 
the equality holds since $\schedule_t$ should be $S$ (otherwise, $\schedule_t = (S \oplus \mathcal{A}_t) \notin \overline{\setS}(t^n)$) and thus $S$ should have the larger weight sum to be chosen by the augmentation algorithm, and the last inequality comes from Lemma~\ref{lem:nearoptimal}.
%
%\iftoggle{tech_report}
%{
    Hence, the third term (\ref{eq:third}) can be upper bounded by
    \begin{equation}  \label{eq:6} 
    \begin{split}
        &\textstyle \sum_{S \in \overline{\setS} (t^n)} \Pr \{\schedule_{t-1}=S\} 
            \cdot \left( 2 \delta \noLink t^{-2} + 1 - \delta \right) \\
        &\textstyle \le 2 \delta \noLink t^{-2} 
            + \left( 1 - \delta \right) \sum_{S \in \overline{\setS} (t^n)} \Pr \{\schedule_{t-1}=S\},
    \end{split}
    \end{equation}
    for all $t \in (t^n,t^{n+1}]$. 

    The result can be obtained by combining (\ref{eq:4}), (\ref{eq:5}), and (\ref{eq:6}).
%}
%{
%    Combining these results, we can obtain (\ref{eq:upper-bound-for-sufficiently-played-matching}).
%}
\end{proof}

In order to apply Lemma~\ref{lem:prob-for-sufficiently-played-matching} to (\ref{eq:tau_samplepath}), we rewrite it in a recursive form. Let $\eta = 1 - \delta$, $G_n = \left( \vert \overline{\setS}(t^n) \vert +\delta \right) \cdot 2\noLink = (n + \delta) \cdot 2\noLink$, and $\Theta_n(t) = \sum_{S \in \overline{\setS}(t^n)} \Pr \{ \schedule_t =S \}$. We have a recursive form of  (\ref{eq:upper-bound-for-sufficiently-played-matching}) as
\[
\Theta_n(t) 
    \leq  \Pr \{\schedule_{t-1} \in \underline{\setS}(t^n) \} + G_n \cdot t^{-2} + \eta \Theta_n(t-1),
\]
for $t \in (t^n, t^{n+1}]$. 
%\iftoggle{tech_report}
%{
    Extending the right side further down to $t^n$, we can obtain that
    \begin{align} 
    \Theta_n(t) 
        &\le \eta^{t-t^n} \Theta_n(t^n) \label{eq:7} \\
        &\textstyle~~+ G_n \sum_{i=t^n+1}^{t} \eta^{t-i} \cdot i^{-2} \label{eq:8} \\
        %&~~+ G_n (t^{-2} +\eta (t-1)^{-2} \cdots +\eta^{t-t^n-1} (t^n+1)^{-2} ) \\
        &\textstyle~~+ \sum_{i=t^n+1}^{t} \eta^{t-i} \cdot \Pr \{\schedule_{i-1} \in \underline{\setS}(t^n) \}. \label{eq:9}
    \end{align}
    By summing it over $t \in (t^n, t^{n+1}]$ on the both sides, we obtain the following lemma.
%}
%{
%    Extending the right side further down to $t^n$ and summing over $t \in (t^n, t^{n+1}]$, we obtain the following lemma.
%}
%
\begin{lemma} \label{lem:theta}
The total number of times that sufficiently played non-near-optimal matchings are selected during $(t^n,t^{n+1}]$ is bounded by 
\begin{equation}\label{eq:reuse}
\textstyle \sum_{t=t^n+1}^{t^{n+1}} \Theta_n(t)
    \leq \frac{1}{\delta}
    ( 
        1+ \frac{\pi^2}{6} G_n  + \mathbb{E} [ \sum_{S \in \underline{\setS}(t^n)} \tau_{S,n+1} ]
    ),
\end{equation}
where $\tau_{S,n+1}$ denote the number of time slots that $S$ is scheduled in $(t^{n},t^{n+1}]$.
\end{lemma}
\noindent The proof of Lemma~\ref{lem:theta} is omitted and can be found in~\cite{tr}.

Now, by taking expectation on (\ref{eq:tau_samplepath}), we can obtain the expected total number of times that non-near optimal matchings are selected up to time $\hat{t}_h (\leq h)$  as 
\begin{align*}
    &\textstyle \sum_{S\in \cnopts_{\vmweight}} \mathbb{E}\left[ \nochosen_S(\hat{t}_h) \right] 
        %=\sum_{S\in \cnopts_{\vmweight}}\sum_{t=1}^{t'} \Pr \{ \schedule_t =S \}\\
        \textstyle = l_h M + \sum_{n=1}^{M-1} \sum_{t=t^n + 1}^{t^{n+1}} \Theta_n(t)\\
        &\textstyle \leq l_hM + \frac{1}{\delta} \sum_{n=1}^{M-1} \left( 1+ G_n \frac{\pi^2}{6} + \mathbb{E}\Big[ \sum_{S \in \underline{\setS}(t^n)} \tau_{S,n+1} \Big] \right),\\
        %&\textstyle \leq l_hM + \left(\frac{1}{\delta} + \frac{\noLink\pi^2}{3}\right) (M-1)
        %+\frac{\noLink\pi^2}{3\delta}\cdot\sum_{n=1}^{M-1} \vert \overline{\setS}(t^n)\vert
        %+\frac{1}{\delta}\cdot \sum_{n=1}^{M-1} \mathbb{E}\left[ \sum_{s=n+1}^{M} \tau_{s,n+1}\right]\\
        %&\textstyle \leq l_hM + (M-1)\left(\frac{1}{\delta} + \frac{\noLink\pi^2}{3}\right) +\frac{\noLink\pi^2}{3\delta}\cdot \frac{M(M-1)}{2} +\frac{1}{\delta}\cdot l_hM\\
        &\textstyle \le l_hM + \frac{M}{\delta} \left( 1 + \frac{\noLink \delta \pi^2}{3} +\frac{\noLink(M-1)\pi^2}{6} +l_h \right).
\end{align*}
%
%\iftoggle{tech_report}
%{
    The last inequality holds since  
    (i) $\sum_{n=1}^{M-1} G_n = \sum_{n=1}^{M-1} (n + \delta) \cdot 2\noLink 
    \le M \cdot \noLink \cdot (2\delta + (M-1))$, and 
    (ii) $\sum_{S \in \underline{\setS}(t^n)} \tau_{S,n+1}$ is the total number that the matchings that have been chosen less than $l_h$ up to $t^n$ are chosen during $(t^n, t^{n+1}]$ and thus results in $\sum_{n=1}^{M-1} \sum_{S \in \underline{\setS}(t^n)} \tau_{S,n+1} \leq \sum_{k=2}^{M} l_h \le l_hM$.
%}{}
%
From $M\leq\vert \setS \vert -1,$ we have 
\begin{align} 
\textstyle \sum_{S\in \cnopts_{\vmweight}} \mathbb{E}\left[ \nochosen_S(\hat{t}_h) \right]
    %\textstyle &\leq \left( 1+\frac{1}{\delta} \right)l_hM +D_1 \\
    %&\leq \left( 1+\frac{1}{\delta} \right) \cdot  \frac{4\noLink^2(\noLink+1)\ln h}{(\alpgapmin)^2} ( \vert \setS \vert -1 ) +D_1,
    &\textstyle \leq D_1 \cdot \frac{\ln h}{(\alpgapmin)^2} \nonumber\\
        &\textstyle \hspace{-1cm} + \frac{\vert\setS\vert-1}{\delta} (1+ \frac{\noLink\delta\pi^2}{3} + \frac{\noLink(\vert\setS\vert-2)\pi^2}{6} ). \label{eq:16}
\end{align}

This provides a bound on the number of times that non-near-optimal matchings are selected up to $\hat{t}_h$. For the rest time $t \in (\hat{t}_h,h]$, we need to compute $\sum_{t=\hat{t}_h+1}^h \Pr \{\schedule_t \in \cnopts_{\vmweight}\}$. Let $S' = \schedule_{t-1} \oplus \mathcal{A}_t$. Since next schedule $\schedule_t$ is either $\schedule_{t-1}$ and $S'$ under the algorithm, we divide the event $\{\schedule_t \in \cnopts_{\vmweight}\}$ into three sub-cases based on $\schedule_{t-1}$ and $S'$, and compute the probability as 
\begin{align}
%&\textstyle \sum_{t=\hat{t}_h+1}^h \sum_{S \in \cnopts_{\vmweight}}  \Pr \{ \schedule_t = S \} \\
\textstyle \Pr \{ \schedule_t \in \cnopts_{\vmweight}\} 
    &\textstyle =  \Pr\{S' \in\cnopts_{\vmweight},\schedule_{t-1}\in\cnopts_{\vmweight}\}\nonumber \\
    &\textstyle \hspace{-0.5cm}+ \Pr\{S' \in \cnopts_{\vmweight},~\schedule_{t-1}\in\nopts_{\vmweight}, \totalmw_{\vmweight}(S')\geq\totalmw_{\vmweight}(\schedule_{t-1})\}\nonumber\\
    &\textstyle \hspace{-0.5cm}+ \Pr\{S' \in \nopts_{\vmweight},~\schedule_{t-1}\in\cnopts_{\vmweight}, \totalmw_{\vmweight}(S')\leq\totalmw_{\vmweight}(\schedule_{t-1})\}.\nonumber
\end{align}
This leads to
\begin{align}
\textstyle \Pr \{ \schedule_t \in \cnopts_{\vmweight}\} 
    &\textstyle \leq  \Pr\{S' \in \cnopts_{\vmweight}~|~\schedule_{t-1} \in\cnopts_{\vmweight}\} \cdot \Pr\{\schedule_{t-1}\in\cnopts_{\vmweight} \} \nonumber \\
    &\textstyle \hspace{-0.5cm}+ \Pr\{\totalmw_{\vmweight}(S')\geq\totalmw_{\vmweight}(\schedule_{t-1}) ~|~S' \in \cnopts_{\vmweight},~\schedule_{t-1}\in\nopts_{\vmweight}\} \nonumber\\
    &\textstyle \hspace{-0.5cm}+  \Pr\{\totalmw_{\vmweight}(S')\leq\totalmw_{\vmweight}(\schedule_{t-1})~|~S' \in \nopts_{\vmweight},~\schedule_{t-1}\in\cnopts_{\vmweight}\}.\nonumber
\end{align}
From Lemma~\ref{lem:probability}, we have $\Pr\{S' \in \cnopts_{\vmweight}~|~\schedule_{t-1} \in\cnopts_{\vmweight}\} = 1 - \Pr\{S' \in\nopts_{\vmweight}~|~\schedule_{t-1} \in \cnopts_{\vmweight}\} \le 1-\delta = \eta$. Since $\hat{\tau}_S(t) \geq l_h$ for all $S$ and $t \in (\hat{t}_h,h]$, Lemma~\ref{lem:nearoptimal} provides an upper bound $2\noLink t^{-2}$ on each conditional probability in the second and the third terms. As a result, we can obtain
\[ \textstyle 
    \Pr \{\schedule_t\in \cnopts_{\vmweight}\} 
        \textstyle \leq  \eta \cdot \Pr\{\schedule_{t-1} \in \cnopts_{\vmweight}\} + 4\noLink t^{-2}.
\]
By extending the inequality in a recursive manner down to $\hat{t}_h$, we obtain that 
\[
\begin{split}
    \Pr \{\schedule_t \in \cnopts_{\vmweight}\} 
        &\textstyle \leq \eta^{t-\hat{t}_h} \cdot \Pr \{ \schedule_{\hat{t}_h}\in \cnopts_{\vmweight} \} 
        + 4\noLink \sum_{i=\hat{t}_h+1}^t \eta^{t-i}i^{-2} \\
        &\textstyle = \eta^{t-\hat{t}_h}  
        + 4\noLink \sum_{i=\hat{t}_h+1}^t \eta^{t-i}i^{-2},
\end{split}
\]
where the last equality holds since $\Pr \{ \schedule_{\hat{t}_h}\in \cnopts_{\vmweight} \} = 1$ from the definition of $\hat{t}_h$. Summing over $t \in (\hat{t}_h, h]$ on the both sides, and from $\eta = 1-\delta$, we have
\begin{equation} \label{eq:19} 
    \textstyle
    \sum_{t=\hat{t}_h+1}^h %\sum\limits_{S \in \cnopts_{\vmweight}} \Pr \{\schedule_t=S\}
    \Pr \{\schedule_t\in \cnopts_{\vmweight}\}
        \le \frac{1-\delta}{\delta}+ \frac{2\noLink\pi^2}{3\delta}.
\end{equation}

Combining (\ref{eq:16}) and (\ref{eq:19}), we obtain
\begin{equation} \label{eq:prop2_key1}
\begin{split}
    &\textstyle \sum_{S\in \cnopts_{\vmweight}} \mathbb{E} [ \nochosen_S(h) ]\\
        &=\textstyle \sum_{S\in \cnopts_{\vmweight}} \mathbb{E}\left[ \nochosen_S(\hat{t}_h) \right]
            + \sum_{t=\hat{t}_h+1}^h \Pr \{\schedule_t\in \cnopts_{\vmweight}\} \\
        &\leq \textstyle  D_1 \cdot \frac{\ln h}{(\alpgapmin)^2} + D_2, 
\end{split}
\end{equation}
where 
%new D_1 and D_2
$D_1 = \left( 1 + \frac{1}{\delta} \right) \cdot 4\noLink^2(\noLink+1) \cdot \left( \vert \setS \vert-1 \right)$, and
$D_2 = \frac{\vert\setS\vert-1}{\delta} (1+ \frac{\noLink\delta\pi^2}{3} + \frac{\noLink(\vert\setS\vert-2)\pi^2}{6} ) + \frac{1-\delta}{\delta} + \frac{2\noLink\pi^2}{3\delta}$.
% old D_1 and D_2
%$D_1 = \frac{\vert\setS\vert-1}{\delta}\left(1+ \frac{\noLink\pi^2\delta}{3}+ \frac{\noLink(\vert\setS\vert-2)\pi^2}{6} \right)$, 
%$D_2=\frac{1-\delta}{\delta}+ \frac{2\noLink\pi^2}{3\delta}$, 
%and $\delta = \min \{ 1, ( \frac{p}{1-p} )^V \} ( \frac{1-p}{k\Sigma} )^V$.

%\iftoggle{tech_report}{}
\textbf{(2) When} $\hat{t}_h > h$ (i.e,. $\exists S$ such that $\nochosen_S(h) < l_h$) : With the same definitions of $l_h$, $\overline{\setS}(t)$, and $\underline{\setS}(t)$, let $\vert \overline{\setS} \vert = |\overline{\setS}(h)|$ and $\vert \underline{\setS} \vert = |\overline{\setS}(h)|$. At this time, we define $\snopts(h)=\{ S^1,S^2,...,S^{\vert\overline{\setS}\vert} \}$ and let $t^n$ denote the time at which matching $S^n$ is sufficiently scheduled, i.e., $\nochosen_{S^n}(t^n)=l_h$. Without loss of generality, we assume $t^1 < t^2<...<t^{\vert \overline{\setS}\vert}$. By time slot $h$, $\underline{\setS}(h)$ is non-empty (since $\hat{t}_h >h)$, and it is clear that $\sum_{S \in \underline{\setS}(h)} \nochosen_S(h) \leq l_h \vert \underline{\setS} \vert$. 

Similar to the case when $\hat{t}_h \le h$, we can obtain 
\[
\begin{split}
&\textstyle \sum_{S \in \cnopts_{\vmweight}} \mathbb{E}\left[ \nochosen_S(h) \right] \\
    &\textstyle =\sum_{S \in \underline{\setS}(h)} \mathbb{E} [ \nochosen_S(h) ]
        %+\sum_{S \in \overline{\setS}(h)} \mathbb{E} [ \nochosen_S(h) ] \\
        + \sum_{t=1}^h\sum_{S \in \overline{S}(h)}\Pr \{ \schedule_t =S \}\\
    %&\textstyle \leq l_h\vert\underline{\setS} \vert + \sum_{t=1}^h\sum_{S \in \overline{S}(h)}\Pr \{ S_t =S \} \\
    &\textstyle \leq l_h\vert\underline{\setS} \vert + l_h\vert \overline{\setS}\vert+ \sum_{n=1}^{\vert \overline{\setS} \vert}\sum_{t=t^n+1}^{t^{n+1}} \Theta_n(t) \\
    &\textstyle \leq l_h M + \sum_{n=1}^{\vert \overline{\setS} \vert} \frac{1}{\delta} 
        %\left( 1 + \frac{\pi^2}{6} G_n + \mathbb{E}\left[\sum_{x=n+1}^{\vert\overline{\setS}\vert} \tau_{x,n+1} \right] \right),
        ( 1 + \frac{\pi^2}{6} G_n + \mathbb{E} [ \sum\limits_{S \in \underline{\setS}(t^n)} \tau_{x,n+1} ] ),
\end{split}
\]
where the last inequality comes from Lemma~\ref{lem:theta}. As in (\ref{eq:16}), we can obtain 
\begin{equation}\label{eq:prop2_key2}
    \begin{split}
    &\textstyle \sum_{S \in \cnopts_{\vmweight}} \mathbb{E}\left[ \nochosen_S(h) \right]\\
    &\le \textstyle l_h M + \frac{\vert \overline{\setS} \vert}{\delta} \left(1 + \frac{\noLink \pi^2 \delta}{3} + \frac{ \noLink (\vert \overline{\setS} \vert + 1) \pi^2}{6} + l_h \right)\\
    %&\leq l_h(\vert \setS\vert-1) + \vert\overline{\setS}\vert\left(\frac{1}{\delta}+\frac{\noLink\pi^2}{3}\right) + \frac{\noLink\pi^2}{6\delta}\vert \overline{\setS}\vert(\vert \overline{\setS}\vert+1) +l_h\frac{\vert \overline{\setS}\vert}{\delta}, 
    %&\textstyle \leq \left( 1+\frac{1}{\delta} \right) \cdot  \frac{4\noLink^2(\noLink+1)\ln h}{(\alpgapmin)^2} ( \vert \setS \vert -1 ) +D_1 +D_2,
    &\leq \textstyle D_1 \frac{\ln h}{(\alpgapmin)^2} + D_2, 
    \end{split}
\end{equation} 
where the last inequality holds due to $\vert \overline{S} \vert \le M-1$.
Proposition~\ref{pro:regretbound} can be obtained by applying (\ref{eq:prop2_key1}) and (\ref{eq:prop2_key2}) to the decomposition inequality (\ref{eq:decomposition}).
%{
%    \textbf{(2) When} $\hat{t}_h > h$ (i.e,. $\exists \setS$ such that $\nochosen_S(h) < l_h$) : In this case, we consider the set of non-near-optimal matchings that are sufficiently scheduled by $h$, and use it in the place of the whole set $\cnopts_{\vmweight}$ of non-near-optimal matching. The proof is similar to the case when $\hat{t}+h \le h$, and can be found in~\cite{tr}.
%}

}

\emph{Remarks:} Despite the logarithmic bound, the algorithm may suffer from slow convergence due to large values of constant $D_1$ and $D_2$. On the other hand, the bound is quite loose because we consider each matching separately. In practice, a link belongs to multiple matchings, and thus it can learn much faster.

%%%%%%%%%%%%%% do not remove the following %%%%%%%%%%%%%%%%%%%%%%%%%%
%\RED{(CJ: Do you think that the following is provable? Some constant may change, but the point is that we have $\frac{1}{\alpgapmin}$ instead of $\frac{1}{(\alpgapmin)^2}$) We note that the proof of Proposition~\ref{pro:regretbound} is obtained by applying the distribution-independent regret decomposition inequality (\ref{eq:decomposition}). By using a distributed-dependent version, e.g., $\alpR(t) \le \sum_{S \in \setS} \mathbb{E} [ \nochosen_S(t) ] \cdot \Delta_S$, where $\Delta_S = \alpha \cdot \maxtotalmw_{\vmweight} - \totalmw_{\vmweight}(S)$, we can have the following result.} \\
%
%\RED{In a network graph of $V$ nodes with maximum degree $\Sigma$, our dAUCB algorithm with seed probability $p$ and parameter $k$ achieves
%\begin{align*}
%\alpR(t) 
%    \leq \left( 1 + \frac{1}{\delta} 
%\right)
%    \cdot  \frac{4\noLink^2(\noLink+1)\ln t}{\alpgapmin} \cdot \left( \vert\setS \vert-1 \right) \\ 
%    ~~~~~+ \alpgapmax \cdot (D_1 + D_2), 
%    ~\text{for all}~t \in \{1,...,\framel\}.
%\end{align*}}
%%%%%%%%%%%%%%%%%%%%%%%%%%%%%%%%%%%%%%%%%%%%%%%%%%%%%%%%%%%%%%%%%%%%%

\subsection{Scheduling efficiency}

We now consider the throughput performance of A$^k$-UCB across multiple frames. It can be obtained through the Lyapunov technique with time unit of frame length. 
\begin{proposition} \label{pro:stability}
For a sufficiently large frame length $\framel$, A$^k$-UCB is rate-stable for any arrival rate strictly inside ${\alpk}\Lambda$. 
\end{proposition}
%
%\iftoggle{tech_report}
\begin{proof}
    %We use the Lyapunov technique with time unit of frame length. 
    Given any $\vlambda$ strictly inside $\alpha\Lambda$ with $\alpha = \alpk$, we consider the Lyapunov function $\lya(t_{\orderF}) = \frac{1}{2} \sum_{i \in \mathcal{L}} \left(q_i(t_{\orderF})\right)^2$ at the start time $t_n$ of the $n$-th frame. If the Lyapunov function has a negative drift for sufficiently large queue lengths, then all the queues will remain finite.
    %, i.e., $\lim_{t\rightarrow \infty} \frac{q_i(t)}{t} = 0$ for all $i$.
    
    From the queue evolution (\ref{eq:q_evolv}), we have 
    \[
    \begin{split}
        q_i(t_{n+1}) 
            &\textstyle \leq \left (q_i(t_{\orderF}) -\sum_{t=t_{\orderF}}^{t_{\orderF}+\framel-1} \serv_i(t) \cdot \mathbb{I}\{ i\in \schedule_t \} \right )^+ \\
            &\textstyle ~~~~+ \sum_{t=t_{\orderF}}^{t_{\orderF} +\framel - 1 } \arv_i(t),
    \end{split}
    \]
    where $\{\schedule_t\}$ denotes the sequence of matchings chosen by A$^k$-UCB. Let $\dri(t_{\orderF}) = L(t_{\orderF+1}) - L(t_{\orderF})$. The drift during a frame time can be written as 
    \[
    \begin{split}
        \mathbb{E} &\left[ \dri(t_n) ~| \vq(t_{\orderF})\right] 
        %&~\textstyle =\frac{1}{2} \mathbb{E}\left[ \sum_{i \in \mathcal{L}} ( q_i(t_{n+1})^2 - q_i(t_{\orderF})^2)  ~\middle| \vq(t_{\orderF}) \right] \\
        \textstyle \leq \frac{1}{2} \sum_{i \in \mathcal{L}} \mathbb{E} [ (\sum_{t=t_{\orderF}}^{t_{\orderF}+\framel-1} \arv_i(t) )^2 | \vq(t_{\orderF}) ]\\  
        &\textstyle +\frac{1}{2}\sum_{i \in \mathcal{L}} \mathbb{E} [ ( \sum_{t=t_{\orderF}}^{t_{\orderF}+\framel-1} \serv_i(t) \cdot \mathbb{I}\{ i\in \schedule_t \} )^2 | \vq(t_{\orderF}) ] \\
        &\textstyle + \sum_{t=t_{\orderF}}^{t_{\orderF}+\framel-1} \mathbb{E} [  \sum_{i \in \mathcal{L}} q_i(t_{\orderF}) \arv_i(t) - \sum_{i \in \schedule_t} q_i(t_{\orderF}) \serv_i(t))  ~| \vq(t_{\orderF}) ],
    \end{split}
    \]
    where the first two terms can be bounded by $C\framel$ for some constant $C$, because $\arv_i(t)$, $\serv_i(t)$, and $\noLink$ are bounded.
    Suppose that we have weight vector $\vmweight$ at time $t_{\orderF}$. Let $S^*$ denote an optimal matchings during the corresponding frame time, i.e., $S^* \in \setOpts_{\mweight} = \argmax_{S \in \mathcal{S}} \sum_{i \in S} \mweight_i$, and let $\maxtotalmw_{\vmweight} = \sum_{i \in S^*} \mweight_i$.
    Since $\vlambda$ strictly inside $\alpha \Lambda$, there exists $\epsilon > 0$ such that $\vlambda + \epsilon \boldsymbol{1} \in \alpha \Lambda$, where $\boldsymbol{1}$ is the vector of all ones. Then from $\mweight_i = \frac{q_i(t_{\orderF})}{\maxq(t_{\orderF})} \mu_i$, we can obtain
    \[
    \begin{split}
        \mathbb{E} &\left[ \dri(t_{\orderF}) ~| \vq(t_{\orderF})\right] \\
            \le&~ C\framel \textstyle + \sum_{t=t_{\orderF}}^{t_{\orderF}+\framel-1} 
                ( \mathbb{E} [ \sum_{i \in \mathcal{L}} q_i(t_{\orderF}) \arv_i(t)  ~| \vq(t_{\orderF}) ] \\
                & \textstyle ~~~~~~~~~~~~~~~~~~~~~~~ - \mathbb{E} [ \sum_{i \in \schedule_t} q_i(t_{\orderF}) \serv_i(t)  ~| \vq(t_{\orderF}) ])\\
            %
            %&\le C\framel \textstyle + \mathbb{E} \left[ \sum_{t=t_{\orderF}}^{t_{\orderF}+\framel-1}  \sum_{i \in \mathcal{L}} q_i(t_{\orderF}) \arv_i(t)~\middle|\vq(t_{\orderF})\right] \\         &\textstyle~~~~~~~~-\mathbb{E}\left[ \sum_{t=t_{\orderF}}^{t_{\orderF}+\framel-1}\alpha \sum_{i \in S^*} q_i(t_{\orderF}) \serv_i(t) ~\middle| \vq(t_{\orderF})\right]\\
            %&\textstyle~~~~~~~~+
            %\mathbb{E} \left[ \sum_{t=t_{\orderF}}^{t_{\orderF}+\framel-1}\alpha \sum_{i \in S^*} q_i(t_{\orderF}) \serv_i(t)  ~\middle| \vq(t_{\orderF}) \right]\\
            %&\textstyle ~~~~~~~~- \mathbb{E} \left[ \sum_{t=t_{\orderF}}^{t_{\orderF}+\framel-1} \sum_{i \in S(t)}~ q_i(t_{\orderF}) \serv_i(t)  ~\middle| \vq(t_{\orderF}) \right] \\
            =&~\textstyle  C\framel + \maxq(t_{\orderF}) \sum_{t=t_{\orderF}}^{ t_{\orderF} + \framel - 1 }  
                \Big( \sum_{i \in \mathcal{L}}\frac{q_i(t_{\orderF})}{\maxq(t_{\orderF})} \lambda_i - \alpha \maxtotalmw_{\vmweight} \Big)\\
                &\textstyle ~~~~~~~~+ \maxq(t_{\orderF}) \sum_{t=t_{\orderF}}^{t_{\orderF}+\framel-1}  
                %\Big( \alpha \sum\limits_{i \in S^*} w_i - \sum\limits_{i \in \schedule_t} w_i \Big)\\
                \Big( \alpha \maxtotalmw_{\vmweight} - \mathbb{E} \left[ \totalmw_{\vmweight}(\schedule_t) ~| \vq(t_{\orderF})\right] \Big)\\
            %&= C\framel \textstyle + \maxq(t_{\orderF}) \Big( \mathbb{E} \left[ \sum_{t=t_{\orderF}}^{t_{\orderF}+\framel-1}  \sum_{e \in \mathcal{L}} \frac{q_i(t_{\orderF})}{\maxq(t_{\orderF})} \lambda_i~\middle|\vq(t_{\orderF})\right] \\ &\textstyle~~~~~~~~~~~~~~~~~~~~-\mathbb{E}\left[ \sum_{t=t_{\orderF}}^{t_{\orderF}+\framel-1}\alpha \sum_{i \in S^*} \mweight_i ~\middle| \vq(t_{\orderF})\right]\\
            %&\textstyle~~~~~~~~~~~~~~~~~~~~+
            %\mathbb{E} \left[ \sum_{t=t_{\orderF}}^{t_{\orderF}+\framel-1}\alpha \sum_{i \in S^*} \mweight_i  ~\middle| \vq(t_{\orderF}) \right]\\
            %&\textstyle~~~~~~~~~~~~~~~~~~~~- \mathbb{E} \left[ \sum_{t=t_{\orderF}}^{t_{\orderF}+\framel-1} \sum_{i \in S(t)}~ \mweight_i  ~\middle| \vq(t_{\orderF}) \right]\Big) \\
            \le&~ C\framel \textstyle - \epsilon \framel \sum_{i \in \mathcal{L}} q_i(t_{\orderF}) + \maxq(t_{\orderF}) \cdot \alpR(\framel).\\
                %&\textstyle+ \maxq(t_{\orderF}) \left( \framel \alpha \cdot \maxtotalmw_{\vmweight} - \mathbb{E} \left[ \sum_{t=t_{\orderF}}^{t_{\orderF}+\framel-1} r_{\vmweight}(S(t))  ~| \vq(t_{\orderF}) \right] \right).
    \end{split}
    %\alpR(t)= t\cdot \alpha \cdot \maxtotalmw_{\vmweight}-\mathbb{E}\left[\sum_{\tau=1}^t r_{\vmweight}(\schedule_{\tau})\right],
    \]
    where the equality holds due to the independence of link rates, and the last inequality holds since $\vlambda + \epsilon \boldsymbol{1} \in \alpha\Lambda$ and thus $\sum_{i \in \mathcal{L}}\frac{q_i(t_{\orderF})}{\maxq(t_{\orderF})} (\lambda_i + \epsilon) < \alpha \maxtotalmw_{\vmweight}$.
    Dividing both sides by $\framel$, we have 
    \[
        \textstyle \frac{1}{T} \mathbb{E} \left[ \dri(t_{\orderF}) ~| \vq(t_{\orderF})\right]
            \le C - \epsilon \sum_{i \in \mathcal{L}} q_i(t_{\orderF})
            + \maxq(t_{\orderF}) \cdot \frac{ \alpR(\framel) }{\framel}.
    \]
    Since Proposition~\ref{pro:regretbound} implies that $\frac{ \alpR(\framel) }{\framel} < \epsilon$ for sufficiently large $\framel$, we have a negative drift for sufficiently large queue lengths.
\end{proof}
Proposition~\ref{pro:stability} means that A$^k$-UCB can stabilize the queue lengths under packet arrival dynamics for any $\vlambda \in \alpk \Lambda$.
%
%{
%    The proof follows the same line of analysis of~\cite{stahlbuhk_adhoc2019}, and uses the Lyapunov technique with a quadratic-form Lyapunov function. Due the limited space, we omit the proof,  which can be found in~\cite{tr}.
%}

\begin{figure*}[t]
    \centering
    \subfigure[A$^k$-UCB]{
        \includegraphics[width=0.31\linewidth]{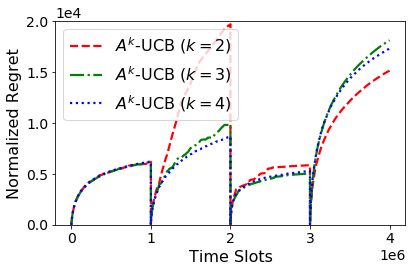}
        \label{fig:global_regret}
    } 
    \subfigure[A$^k$-UCB ($\alpha$-regret)]{
        \includegraphics[width=0.31\linewidth]{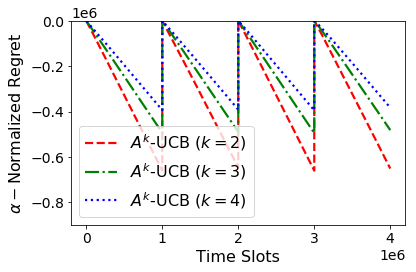}
        \label{fig:alpha_regret}
    } 
    \subfigure[$d$A$^k$-UCB]{
        \includegraphics[width=0.31\linewidth]{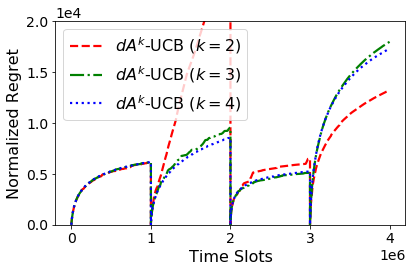}
        \label{fig:local_regret}
    }
    \caption{Regret traces for learning performance. For comparison across frames, the regret is reset to $0$ at each frame boundary ($\framel = 10^6$ time slots), and normalized by the maximum expected reward $\maxtotalmw_{\vmweight}$.}
    \label{fig:regret}
\end{figure*}

\begin{figure*}
    \centering
    \subfigure[A$^k$-UCB ($k = 3$) with different $\framel$'s]{
        \includegraphics[width=0.31\linewidth]{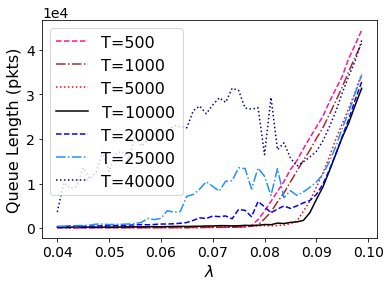}
        \label{fig:frame_size}
    }
    \subfigure[A$^k$-UCB (zoom-in)]{
        \includegraphics[width=0.31\linewidth]{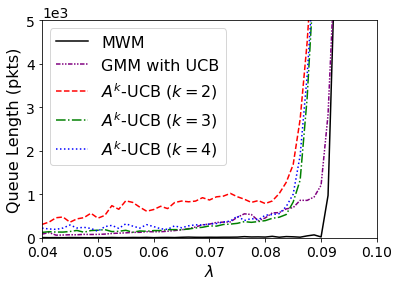}
        \label{fig:global_queue}
    }
    \subfigure[$d$A$^k$-UCB (zoom-in)]{
        \includegraphics[width=0.31\linewidth]{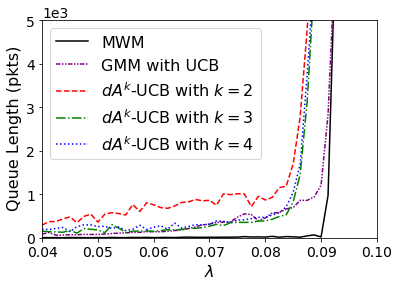}
        \label{fig:local_queue}
    }
    \caption{Queue lengths for scheduling efficiency. Given a scheduler, if the arrival rate gets closer to the boundary of its stability region, the queue length soars quickly.} \label{fig:stability}
\end{figure*}

%%%%%%%%%%%%%%%%%%%%%%%%%%%%%%%%%%%
%\section{Extensions of \BLUE{A$^k$-UCB}}
%%%%%%%%%%%%%%%%%%%%%%%%%%%%%%%%%%%
\section{Distributed implementation ($d$A$^k$-UCB)} \label{sec:distributed}

Although the augmentation algorithm has $O(k)$ complexity and amenable to implement in a distributed fashion, the link index $\viweight_t$ of A$^k$-UCB includes \emph{global information} $\maxq(t_\orderF)$ -- the largest queue length in the network at the start of each frame $\orderF$. This normalization is due to Hoeffding inequality and essential for the provable regret performance bound. 
%The same problem has been also observed in the greedy algorithm of~\cite{stahlbuhk_adhoc2019}.

We pay attention to the fact that A$^k$-UCB indeed learns the expected value of the queue weighted link rate, i.e., $q_i(t_\orderF) \mu_i$, and the global information $\maxq(t_\orderF)$ takes the role of normalizing the weight in the range of $[0,1]$. 
This implies that it may be able to separate the normalizing parameter from the learning. To this end, we develop a distributed version of A$^k$-UCB, denoted by $d$A$^k$-UCB, and describe it with two key differences. For the ease of exposition, we assume that $\mathcal{A}_t$ consists of a single augmentation.
\begin{enumerate} 
    \item {\bf Local normalizer:} Each node $v$ maintains a local normalizer $\tilde{q}_v$, which is initialized to $\max_{u \in \mathcal{N}(v)} q_{(u,v)}(t_\orderF)$ at the beginning of each frame $n$. At each time slot $t$, node $v$ in an augmentation updates its local normalizer twice as follows. 1) In each initialization stage or path augmenting stage, the REQ message from $u$ to $v$ includes additional information of $\tilde{q}_u$. The receiving node $v$ sets $\tilde{q}_v \leftarrow \max\{ \tilde{q}_u, \tilde{q}_v \}$. This repeats while building the augmentation. After the cycle-checking stage, the terminus $w$ has $\tilde{q}_w = \tilde{q}^*$ that is the largest local normalizer in the augmentation. 2) In the back-propagating stage, this value $\tilde{q}^*$ is back-propagated together and each node $v$ in the augmentation sets $\tilde{q}_v \leftarrow \tilde{q}^*$. Hence, at the end of the time slot, all the nodes in the augmentation have the same local normalizer~$\tilde{q}^*$.
    \item {\bf Separate gain computation:} In the meantime, we change the way to compute the gain. To elaborate, let $G'_u$ denote the new gain normalized by $\tilde{q}_u$. At each mini-slot in the path augmenting stage, whenever node $u$ transmits an REQ message to node $v$, we divide $G'_u$ into $G'_{u,1} + G'_{u,2}$, where $G'_{u,1}$ is for average reward (normalized by factor $\tilde{q}_u$) and $G'_{u,2}$ for confidence interval. They are included in the REQ message, separately. Then, after the receiving node~$v$ updates the local normalizer $\tilde{q}_v$, it re-normalizes the received reward gain as $G'_{u,1} \cdot \tilde{q}_u / \tilde{q}_v$. Once the next link is decided as $i = (v,n)$, it computes $G'_{v,1}$ accordingly by either adding or subtracting its average reward normalized by~$\tilde{q}_v$, i.e., $\sweight'_i (t) = \frac{q_i(t_\orderF)}{\tilde{q}_v} \cdot \frac{1}{\nochosen_i(t)}\sum_{j=t_n + 1}^{t}  \serv_{i}(j) \cdot  \mathbb{I}\{i \in S_j\}$. $G'_{u,2}$ can be obtained simply by adding the confidence interval. Let $A'$ is the augmentation up to node $v$, and let $A'_1=A'-\schedule_{t-1}$ and let $A'_2=A' \cap \schedule_{t-1}$. Then the gains
    \[
    \begin{split}
        G'_{v,1} 
            &\textstyle = \sum_{i \in A'_1} \hat{w}'_i - \sum_{j \in A'_2} \hat{w}'_j \\ 
            &\textstyle = \frac{\tilde{q}_u }{\tilde{q}_v} G'_{u,1} + (\mathbb{I}\{v\in A'_1\} - \mathbb{I}\{v\in A'_2\}) \cdot \hat{w}'_v,\\
        G'_{v,2} 
            &\textstyle = \sum_{i \in A'_1} \sqrt{\frac{(\noLink+1)\ln t}{\nochosen_i}} - \sum_{j \in A'_2} \sqrt{\frac{(\noLink+1)\ln t}{\nochosen_j}} \\
            &\textstyle = G'_{u,2} + (\mathbb{I}\{v\in A'_1\} - \mathbb{I}\{v\in A'_2\}) \cdot \sqrt{\frac{(\noLink+1)\ln t}{\nochosen_v}},
    \end{split}
    \]
    can be computed given the value of $\tilde{q}_u$ and the gains $G'_{u,1}, G'_{u,2}$. By repeating this during the augmenting stage, we can obtain the gain  normalized by~$\tilde{q}^*$ at the terminus.
\end{enumerate}
\noindent \emph{Remarks:} During a frame time, the local normalizer of a node is non-decreasing over time slots. In addition, at the same time slot, two nodes in the network may have a different normalizer value. Hence, our previous analysis results for A$^k$-UCB cannot be directly applied to $d$A$^k$-UCB. However, we highlight that, given a time slot, all the nodes in the same augmentation have the same value of the (local) normalizer, which is of importance, since the gain comparison for making a decision occurs only within an augmentation. 
%We conjecture that $d$A$^k$-UCB also achieves the logarithmic growth of $\alpha$-regret and $\alpha$ fraction of the capacity region, with $\alpha=\alpk$. A rigorous proof for this remains as an interesting open problem. 
On the other hand, as the time slot $t$ increases, the value of the global normalizer $\maxq(t_n)$ is disseminated throughout the network and all the local normalizers will converge to this value. 
Considering that, it is not difficult to show that there exists some $T'$ such that all nodes~$v$ have $\tilde{q}_v = \maxq(t_n)$ with probability close to $1$ for all $t > T'$, we believe that $d$A$^k$-UCB also achieves $O(\log T)$ regret performance and ${\alpk}\Lambda$ capacity, if the frame length $\framel$ is sufficiently large. Rigorous proof remains as future work.

\section{Numerical Results} \label{sec:sim}

We evaluate the performance of our proposed schemes through simulations. We first consider a 4x4 grid network topology with the primary interference model, and then conduct extend simulations with a randomly generated network.
Time is slotted. At each time slot, a packet arrives at link $i$ with probability $\lambda_i$, and for a scheduled link $j$, a packet successfully departs the link with probability $\mu_j$. Both $\vlambda$ and $\vmu$ are unknown to the controller. The arrivals and the departures are independent across the links and time slots.
%the instance link rate $\dep_i(t)$ is drawn from an i.i.d. Bernoulli distribution with mean $\mu_i$.}
%$\lambda_i = 0.08$ for all $i$
%$\mu_i \in [0.25, 0.75]$

\begin{figure*}
    \centering
    \subfigure[A randomly generated network]{
        \includegraphics[width=0.31\linewidth]{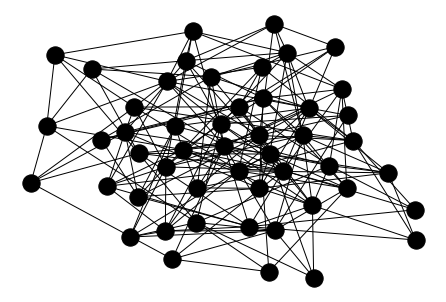}
        \label{fig:mTopology}
    }
    \subfigure[Performance in the random network]{
        \includegraphics[width=0.31\linewidth]{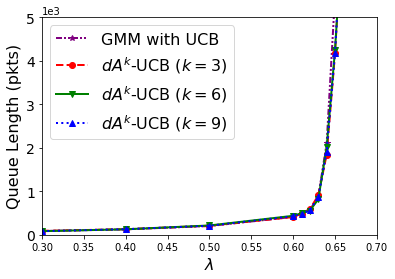}
        \label{fig:mdAkUCB}
    }
    \subfigure[Performance in a ring network]{
        \includegraphics[width=0.31\linewidth]{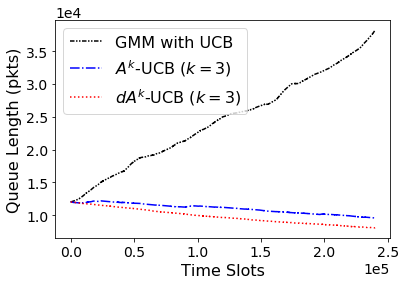}
        \label{fig:ring}
    }
    \caption{Stability in different network topologies. $d$A$^k$-UCB achieves high performance in a larger randomly-generated network, and GMM with UCB may suffer from low performance in a specific ring network.} \label{fig:mstability}
\end{figure*}

%%% Regret performance %%%
\noindent {\bf Regret performance:} We first investigate the regret performance of A$^k$-UCB and $d$A$^k$-UCB. We set the seed probability $p=0.2$ for the schemes. We consider a large frame length of $\framel = 10^6$ time slots to observe their regret growth. An identical arrival rate $\lambda_i = 0.08$ is set for all links $i$, and the departure rate $\mu_i$ is set uniformly at random in range $[0.25, 0.75]$\footnote{Due to randomized $\mu_i$, each link has different traffic load despite the identical arrival rate. Setting a different arrival rate for each link leads to similar results.}. We simulate the two schemes with different~$k$'s and measure their regret performance. For the comparison across frames, the regret value is set to $0$ at each frame start, and normalized with respect to the maximum expected reward sum $\maxtotalmw_{\vmweight}$ within the frame.

Fig.~\ref{fig:global_regret} illustrates the regret traces of A$^k$-UCB, which is an average of 10 simulation runs. We can observe the logarithmic regret growth in both cases. Recall that our regret analysis in Proposition~\ref{pro:regretbound} is for $\alpha$-regret with $\alpha = \alpk$. Thus empirical performance of the proposed schemes are much better than the analytical bound. The performance in terms of $\alpha$-regret is shown in Fig.~\ref{fig:alpha_regret}, where the gap from $0$ can be interpreted as the level of practical difficulty in achieving analytic performance bound: as $k$ increases, it harder to achieve $\alpk$-regret. 
Fig.~\ref{fig:local_regret} shows the regret values of $d$A$^k$-UCB. Comparing with those of A$^k$-UCB, they achieve similar learning performance in terms of regret, and thus we can conclude that the performance loss due to local normalizer in $d$A$^k$-UCB is negligible in practice.

%%% queue lengths %%% 
\noindent {\bf Scheduling efficiency:} 
We evaluate scheduling efficiency of the proposed schemes in the grid network.
We use the same simulation settings, but change frame length $\framel$ and arrival rate $\lambda$ ($\lambda_i = \lambda$ for all $i$). By increasing $\lambda$, the arrival rate gets closer to the boundary of its stability region, and when this occurs, the queue length will soar quickly. We conduct each simulation for $10^6$ time slots and measure queue lengths when the simulations end. For each $\lambda$, we average the queue length of $10$ simulation runs. By observing the arrival rate where the queue length starts increasing quickly, we can indirectly compare the achievable stability region $\gamma\Lambda$ for different scheduling policies~\cite{joo_ton09a}. A policy with larger $\gamma$ is better.
%change the arrival rate by setting $\lambda_i = \lambda$ for all $i$ and increasing the  value of $\lambda$. We consider the same simulation settings with different $\framel$'s and $\lambda$'s. Note that, under any scheduling scheme, when the arrival rate gets closer to the boundary of its stability region, the queue length soars quickly. Thus, we can estimate the performance bound of the scheme from the queue lengths. 

Fig.~\ref{fig:frame_size} demonstrates how the bound changes according to the frame length. From the results, we can observe that the critical point of $\lambda$, around which the queue length starts soaring, is increasing for $\framel \le 10^6$ and then decreasing for $\framel \ge 2 \cdot 10^6$. This is somewhat expected, since too small frame length will lead to incomplete learning, and a larger frame length results in a relatively slower response to the queue dynamics.

We now evaluate the performance of A$^k$-UCB in comparison with two schemes of MWM and UCB-based GMM. MWM is a well-known optimal scheduler~\cite{tassiulas_tac92}, and it is a centralized algorithm that not only requires the knowledge about the weight $q_i(t) \mu_i$ at each time slot $t$, but also has a high-order computational complexity. In our simulations, we use its performance as a reference value. The UCB-based GMM~\cite{stahlbuhk_adhoc2019}) finds a matching by including the link with the highest UCB index first. It is known to achieves $\frac{1}{2}\Lambda$ and has the linear computational complexity.

Fig.~\ref{fig:global_queue} demonstrate the queue lengths of MWM, UCB-based GMM, and A$^k$-UCB. MWM operates at each time slot, and for UCB-based GMM and A$^k$-UCB, we use $\framel = 5000$. Each simulation runs for $10^6$ time slots (i.e., $200$ frames) and we measure the queue lengths after the simulations.
For MWM, the queue length quickly increases at around $\lambda = 0.09$, which can be considered as the boundary of the capacity region $\Lambda$. Under A$^k$-UCB with $k=2,3,4$, the queue lengths increase quickly around $\lambda = 0.084$ for all $k$, which exceeds their theoretic bound $\alpk \cdot 0.09 = 0.03, 0.045, 0.054$, respectively. Interestingly, the impact of $k$ on throughput is not significant, which seems to be due to the small network size -- we can observe that a larger $k$ leads to lower queue lengths in all arrival rates and thus achieves better delay performance. The UCB-based GMM achieves the performance closet to that of MWM, which is also far beyond its theoretic bound $\frac{1}{2}$.
Fig.~\ref{fig:local_queue} shows that $d$A$^k$-UCB achieves almost the same performance as A$^k$-UCB.

\noindent {\bf Performance in randomly generated networks:} 
We now evaluate the performance of the proposed schemes in a larger, irregular-shaped network. To this end, we randomly generate a network of $50$ nodes and $200$ links as shown in Fig.~\ref{fig:mTopology}, and run simulations for $1000$ frame times with $\framel = 500$ (i.e., total $5\cdot 10^5$ time slots). For each link $i$, we set the successful transmission rate $\mu_i$ uniformly at random in range $[0.25, 0.75]$, and set the arrival rate as $\lambda_i = \lambda \cdot \rho_i$, where $\rho_i$ is chosen uniformly at random in range $[0.4,0.7]$. Due to high computational complexity of MWM, we simulate only UCB-based GMM, A$^k$-UCB, and $d$A$^k-$UCB in this experiment, and use the performance of UCB-based GMM as a reference value. It has been observed that GMM algorithm often achieves the optimal scheduling performance in this randomized network environment~\cite{joo_ton09a}.

Fig.~\ref{fig:mdAkUCB} demonstrates the queue lengths of UCB-based GMM and $d$A$^k$-UCB with $k=3,6,9$. All $4$ schemes achieve almost-identical performance, and the setting of $k$ is not sensitive to the performance. The results of A$^k$-UCB are almost identical to that of $d$A$^k$-UCB, and thus omitted.

\noindent {\bf Low performance of UCB-based GMM:} 
So far, UCB-based GMM achieves close-to-optimal performance despite its low performance guarantee of $\frac{1}{2}\Lambda$. It is an interesting question whether the performance bound is not tight due to technical difficulties and its true performance is close to the optimal.
%Then one may use UCB-based GMM in moderate-size networks accounting for its linear complexity.
Unfortunately, however, we shows in the next experiment that this is not the case, and UCB-based GMM may suffer from low performance in a certain scenario. 
%\iftoggle{tech_report}
%{
    We consider a $6$-link ring topology, where the links are numbered from $1$ to $6$ in a clockwise direction. The service rate of each link follows a Bernoulli distribution with mean $\frac{1}{2}$ and the packet arrival on each link is also a Bernoulli process with mean\footnote{This implies that $\vlambda \in (\frac{2}{3}+4\epsilon)\Lambda$ since an optimal scheduler can support arrival rate of up to $\frac{1}{4}$ on each link.} $\frac{1}{6}+\epsilon$ where $\epsilon = 0.08$. 
    We set the frame length $\framel = 6000$. Other environment settings are the same as before, except that the initial queue length is $\{\frac{3\framel}{6}, \frac{2\framel}{6}, \frac{\framel}{6},\frac{3\framel}{6}, \frac{2\framel}{6}, \frac{\framel}{6}\}$.
%}
%{
%    Detailed settings for the simulation can be found in~\cite{tr}.
%}
Fig.~\ref{fig:ring} shows the queue length traces for UCB-based GMM and our proposed schemes. We can observe that while the queue lengths of A$^k$-UCB and $d$A$^k$-UCB are stabilized, those of the UCB-based GMM keep increasing.
%\iftoggle{tech_report}
%{
    This is because, the greedy algorithm tends select a matching with the two links of the largest queue at the beginning of each frame. In contrast, A$^k$-UCB and $d$A$^k$-UCB select a matching with three links by considering their weight sum.
%}
This result implies that in a certain circumstance, UCB-based GMM may suffer from low scheduling efficiency. 

\begin{comment}
\begin{figure}[t]
\includegraphics[width=0.35\textwidth]{Ring_GMM.png}
\caption{Stability in a 6-link ring topology.}
\label{fig:ring}
\end{figure}
\end{comment}

\section{Conclusion} \label{sec:conclusion}

In this work, we addressed the joint problem of learning and scheduling in multi-hop wireless networks. Without a priori knowledge on link rates, we aim to find a sequence of schedules such that all the queue lengths remain finite under packet arrival dynamics. By incorporating the augmentation algorithm into a learning procedure, we develop provably efficient low-complexity schemes that i) achieve logarithmic regret growth in learning, and ii) have the throughput performance that can be arbitrarily close to the optimal. We extend the result to a distributed scheme that is amenable to implement in large-scale networks. We also verify our results through simulations.

\bibliographystyle{IEEEtran}
\bibliography{bibliography}

\iftoggle{tech_report}
{}
{

\begin{IEEEbiography}[{\includegraphics[width=0.96in, keepaspectratio]{pdh2.eps}}]
{Daehyun Park} received his M.S. degree from the school of ECE at Ulsan National Institute of Science and Technology (UNIST) in 2020. His research interests include multi-armed bandits.
%\vspace{\biovspace}
\end{IEEEbiography}

\begin{IEEEbiography}[{\includegraphics[width=0.96in, keepaspectratio]{sunjung_1.eps}}]
{Sunjung Kang} received her M.S. degree from the school of ECE at Ulsan National Institute of Science and Technology (UNIST) in 2018. She is currently a Ph.D. student in the department of ECE at The Ohio State University. Her research interests include the age of information, remote estimation and multi-armed bandits.
%\vspace{\biovspace}
\end{IEEEbiography}

%\vspace{\biovspace}
\begin{IEEEbiography}[{\includegraphics[width=0.96in,keepaspectratio]{cjoo-20190415.eps}}]
{Changhee Joo} received the Ph.D. degree from Seoul National University in 2005. He was with Purdue University and The Ohio State University, and then worked at Korea University of Technology and Education (KoreaTech), and  Ulsan National Institute of Science and Technology (UNIST). Since 2019, he has been with Korea University. 
His research interests are in the broad areas networking, learning, modeling, and optimization. He was a recipient of the IEEE INFOCOM 2008 Best Paper Award, the KICS Haedong Young Scholar Award (2014), the ICTC 2015 Best Paper Award, and the GAMENETS 2018 Best Paper Award. 
He was an Associate Editor of the \emph{IEEE/ACM Transactions on Networking}, and currently an Editor of the \emph{IEEE Transactions Vehicular Technology}, a Division Editor of the \emph{Journal of Communications and Networks}, and has served several primary conferences as a technical program committee member, including IEEE INFOCOM, ACM MOBIHOC, IEEE WiOpt, and IEEE GLOBECOM.
%\vspace{\biovspace}
%\vspace{-0.4cm}
\end{IEEEbiography}
\vfill
}

%%%%%%%%%%%%%%%%%%%%%%%%%%%%
% Appendix
\appendices

%%%%%%%%%%%%%%%%%%%%%%%%%%%%
% app A & B
\iftoggle{tech_report}  % app:A and app:B are included in Technical Report only
{

%%%%%%%%% app:A and app:B need a revision !!!!!!!

\section{Proof of Lemma~\ref{lem:possibility}} \label{app:A}

The lemma shows the possibility for our algorithm to yield a disjoint augmentation of $\schedule_{t-1}$ satisfying the inequality. To this end, we consider the best augmentation set that the algorithm can generate. It can be obtained by combining the previous schedule $\schedule_{t-1}$ and an optimal matching $\opts_{\vmweight}$ with the largest weight sum.

We consider two specific matchings (i.e., super arms) of $\schedule_{t-1}$, the schedule at the previous time slot, and $\opts_{\vmweight} \in \setOpts_{\vmweight}$, an optimal matching with the largest weight sum under weight $\vmweight$. Let us define symmetric difference $S_d = \schedule_{t-1}~ \Delta ~ \opts_{\vmweight} = (\schedule_{t-1}-\opts_{\vmweight}) \cup (\opts_{\vmweight} - \schedule_{t-1})$, and consider graph $G' =(V,S_d)$ that contains only the links in $S_d$. Note that for any vertex in $G'$, its degree is at most $2$ because $\schedule_{t-1}$ and $\opts_{\vmweight}$ are a matching. Further, we can define \emph{component} as a set of connected links in $G'$, which is either a path or an even-length cycle such that links in $\schedule_{t-1}$ and links in $\opts_{\vmweight}$ are alternating, i.e., an augmentation. There can be multiple components in $G'$. For a component $C$, let $C_{opt} = C \cap \opts_{\vmweight}$ and $C_{t-1} = C \cap \schedule_{t-1}$. We note that $size(C) =\vert C_{opt} \vert$ and both $C_{opt}$ and $C_{t-1}$ are a matching in $G$.

Given weight $\vmweight$, let $G_t^{\vmweight} (A)$ denote the gain of augmentation defined as in (\ref{eq:gain}). We also extend the definition to a set $\mathcal{A}$ of disjoint augmentation by adding up the gains, i.e., $G_t^{\vmweight} (\mathcal{A}) = \sum_{A \in \mathcal{A}} G_t^{\vmweight} (A)$. The following two lemmas show that each component $C$ can be decomposed into small-size augmentations with a set of large-weight-sum links. 
\begin{lemma} \label{lem:path}
For component $C$ that is a path, there exists a set of disjoint augmentations $\mathcal{A}(C) \subset C$ such that 
\begin{enumerate}
    \item $size(A) \leq k$ for all augmentations $A \in \mathcal{A}(C)$,
    \item $G^{\vmweight}_t(\mathcal{A}(C)) \geq \frac{k}{k+1} \cdot \totalmw_{\vmweight}(C_{opt}) - \totalmw_{\vmweight}(C_{t-1})$.
\end{enumerate} 
\end{lemma}

\begin{proof}
When $size(C) \le k$, we set $\mathcal{A}(C) = C$. The gain, by definition, equals $G^{\vmweight}_t (\mathcal{A}(C))= \totalmw_{\vmweight}(C_{opt})-\totalmw_{\vmweight}(C_{t-1})$. Thus the two conditions are satisfied. 

When $size(C) > k$, we construct a family of augmentation sets $\{\mathcal{A}_i\}$ and show that at least one of the sets satisfies the two conditions. For path $C$, we select an endpoint. From the endpoint, we denote the first $(k+1)$ links in $C_{opt}$ by $e_1, e_2, \dots, e_{k+1}$. For each $e_i$, we can construct an augmentation set as follows. From $C$, we remove $e_i$ and every $2(k+1)$-th link thereafter (until we reach the other endpoint). Since $C$ is an augmentation and $e_i \in C_{opt}$, any subsequently removed links also belong to $C_{opt}$. After the removals, augmentation $C$ is divided into a set of disjoint augmentations, each of which includes at most $k$ links of $C_{opt}$. Let $\mathcal{A}_i$ denote this set of disjoint augmentations. By repeating the procedure for each $1 \le i \le k+1$, we can construct a family of augmentation sets $\{\mathcal{A}_i\}_{i=1}^{k+1}$ such that any augmentation $A \in \cup_{i=1}^{k+1} \mathcal{A}_i$ satisfies $size(A) \le k$.

Note that, in the construction procedure of $\{\mathcal{A}_i\}_{i=1}^{k+1}$, each link of $C_{opt}$ is removed exactly once, and no link of $C_{t-1}$ is removed. This implies that the gain sum over $\{\mathcal{A}_i\}$ satisfies
$\sum_{i=1}^{k+1} G^{\vmweight}_t (\mathcal{A}_i) = k \cdot \totalmw_{\vmweight}(C_{opt}) - (k+1) \cdot \totalmw_{\vmweight}(C_{t-1})$. Hence, there exists at least one $j$ such that 
$G^{\vmweight}_t (\mathcal{A}_j) \geq \frac{k}{k+1} \cdot \totalmw_{\vmweight}(C_{opt})
         -\totalmw_{\vmweight}(C_{t-1})$.
\end{proof}

Lemma~\ref{lem:path} shows the existence of a good set of disjoint augmentations in every path component in graph $G'$. We use Lemma~\ref{lem:path} to prove a slightly weaker result for a cycle component in $G'$.
\begin{lemma} \label{lem:cycle}
For component C that is a cycle, there exists a set of disjoint augmentations $\mathcal{A}(C)$ such that
%If component $C$ of $S_d$ is a cycle, then there exists a set $\mathcal{A}_C$ of disjoint augmentations such that 
    \begin{enumerate}
        \item $size(A) \leq k$ for all augmentations $A \in \mathcal{A}(C)$,
        \item $G^{\vmweight}_t (\mathcal{A}(C)) \geq \alpk \cdot  \totalmw_{\vmweight}(C_{opt})-\totalmw_{\vmweight}(C_{t-1})$.
    \end{enumerate}
\end{lemma}
%\indent Note that the ratio is $\alpk$ instead of $\frac{k}{k+1}$ as it was for paths.
\begin{proof} 

As in the proof of Lemma~\ref{lem:path}, when $size(C) \le k$, we can set $\mathcal{A}(C) = C$, and the two conditions are satisfied. Thus, we assume that $size(C) > k$.

Let $e \in C_{opt}$ denote the link with the smallest weight in $C_{opt}$. Consider path $\hat{C}=C-e$ and define $\hat{C}_{opt} = \hat{C}\cap \opts_{\vmweight}$ and $\hat{C}_{t-1} = \hat{C}\cap \schedule_{t-1}$. Since $\hat{C}_{opt}$ is obtained by removing the link of the smallest weight from $C_{opt}$, it has a larger average weight than $C_{opt}$, i.e., $\frac {\totalmw_{\vmweight}(\hat{C}_{opt})}{\vert C_{opt} \vert-1} \ge \frac{\totalmw_{\vmweight}(C_{opt})}{\vert C_{opt} \vert}$, which implies that, from $\vert C_{opt} \vert = size(C) > k$,
\[
\totalmw_{\vmweight}(\hat{C}_{opt}) 
    %\geq \frac{\vert C_1\vert-1}{\vert C_1\vert} \cdot \totalmw_{\vmweight}(C_{1}) 
    \textstyle  \geq \frac{k}{k+1} \cdot \totalmw_{\vmweight}(C_{opt}).
\]

Lemma~\ref{lem:path} tells us that there exists a set $\mathcal{A}(\hat{C})$ of disjoint augmentations, where i) the size of each augmentation is at most $k$, and ii) the following inequality holds: $G^{\vmweight}_t(\mathcal{A}(\hat{C})) 
        \textstyle \geq \frac{k}{k+1} \cdot  \totalmw_{\vmweight}(\hat{C}_{opt})-\totalmw_{\vmweight}(\hat{C}_{t-1})
        %&\geq \frac{k}{k+1}\left(\frac{k}{k+1}\totalmw_{\vmweight}(C_{1})\right)-\totalmw_{\vmweight}(\hat{C}_{2})\\
        \textstyle \geq \alpk \cdot \totalmw_{\vmweight}(C_{opt})-\totalmw_{\vmweight}(C_{t-1})$.
\end{proof}

We now build $\augstar$ for the proof of Lemma~\ref{lem:possibility}. For each component $C$ of $S_d$, the corresponding $\mathcal{A}(C)$ can be specified by Lemma~\ref{lem:path} or Lemma~\ref{lem:cycle}. We set $\augstar = \cup_C \mathcal{A}(C)$. Then all augmentations $A \in \augstar$ have $size(A) \le k$. Further, we have 
\[
\begin{split}
    G^{\vmweight}_t(\augstar)
        &\textstyle \geq \alpk \cdot  \sum_C \totalmw_{\vmweight}(C_{opt}) 
            -\sum_C \totalmw_{\vmweight}(C_{t-1})  \\
        &\textstyle \geq \alpk \cdot \sum_{i \in \opts_{\vmweight}} \mweight_i 
            -\sum_{i \in \schedule_{t-1}} \mweight_i.
\end{split}
\]
Hence, we have 
$\totalmw_{\vmweight} (\schedule_{t-1} \oplus \augstar) = G^{\vmweight}_t (\augstar) + \totalmw_{\vmweight} (\schedule_{t-1}) \ge \alpk \cdot \maxtotalmw_{\vmweight}$.

\section{Proof of Lemma~\ref{lem:probability}} \label{app:B}

Suppose that, given $\schedule_{t-1}$, our algorithm has generated the set $\mathcal{A}$ of disjoint augmentations. Let $\augstar$ denote a near optimal set of augmentations satisfying  $\totalmw_{\vmweight}(\schedule_{t-1}\oplus\mathcal{A^*}) \ge \alpk \cdot \totalmw_{\vmweight} (\opts_{\vmweight})$, which consists of $d$ disjoint augmentations $\{A_1, \dots, A_d\}$. For each $A_i$, let $n_i$ denote an (arbitrary) endpoint of $A_i$ if $A_i$ is a path, or an arbitrary node in $A_i$ if $A_i$ is a cycle.
%We obtain a lower bound to $\Pr[\mathcal{A}=\mathcal{A}^* ~ \vert ~\schedule_{t-1}]$.
Consider an event $\boldsymbol{\mathcal{E}}$ under our algorithm such that all of the following are true: 
\begin{itemize}
\item the set of self-selected seeds equals $\{n_1, \dots, n_d\}$,
\item all seed nodes $n_i$ select its intended size as $size(A_i)$, and
\item the generated set of augmentation $\mathcal{A}$ equals $\mathcal{A}^*$. 
\end{itemize}
Since it is clear that $\Pr[\mathcal{A}=\mathcal{A}^* \vert \schedule_{t-1}] \geq \Pr[\boldsymbol{\mathcal{E}} \vert \schedule_{t-1}]$, we focus on a lower bound of $\Pr[\boldsymbol{\mathcal{E}} \vert \schedule_{t-1}]$.

Note that (i) the seed selection of $\{n_i\}$ occurs with probability $p^d(1-p)^{V-d}$, where $V = \vert \mathcal{V}\vert$, 
and (ii) the probability that all seeds $n_i$ independently select its intended size of $size(A_i)$ is $k^{-d}$. 
(iii) Finally, we compute the probability that the generated set of augmentation $\mathcal{A}$ equals $\mathcal{A}^*$, which occurs when, at each iteration, each active node select its corresponding neighbor that belongs to $\mathcal{A}^*$. Each selection has probability $\Sigma^{-1}$, where $\Sigma$ is the maximum node degree, and total $V$ selections will be made under the algorithm. Thus, combining the three probabilities, we have that 
\begin{equation}
\begin{split}
\Pr [\boldsymbol{\mathcal{E}} \vert \schedule_{t-1}] 
    &\geq p^d(1-p)^{V-d} \cdot k^{-d} \cdot \Sigma^{-V} \\
    &\textstyle \geq \min \left\{ 1,\left( \frac{p}{1-p} \right)^V \right\} \left( \frac{1-p}{k\Sigma}\right)^V.
\end{split}
\end{equation}
Setting $\delta = \min \left\{ 1,\left( \frac{p}{1-p} \right)^V \right\} \left( \frac{1-p}{k\Sigma}\right)^V$ completes the proof.
}
%%%%%%%%%%%%%%%%%%% end of app:A and app:B

%%%%%%%%%%%%%%%%%%%%%%%%%%%%%
% app:C
%\iftoggle{tech_report}{   % app:C will be included in Technical Report only

\iftoggle{tech_report}
{
%revision by 190324
\section{Proof of Lemma~\ref{lem:nearoptimal}} \label{app:C}

At time $t>0$, define $l_t=\lceil \frac{4\noLink ^2(\noLink +1)\ln t}{(\alpgapmin)^2} \rceil$. Suppose that a non-near-optimal matching $S \in \cnopts_{\vmweight}$ such that $\nochosen_S(t) \ge l_h$ is scheduled. It also implies $\nochosen_{i}(t) \geq l_h$ for all links $i \in S$. Also, consider an arbitrary near-optimal matching $S' \in \nopts_{\vmweight}$, and let $S = \{ e_1, e_2, \dots, e_A \}$ and $S' = \{ e'_1, e'_2, \dots, e'_B \}$, where $A = \vert S \vert$ and $B=\vert S' \vert$, respectively.

We consider the event that the total index sum over links in $S$ is greater than that over links in $S'$. 
\begin{equation} \label{eq:event_reward}
\begin{split}
    &\mathbb{I} \left\{\totalmw_{\viweight_t}(S) \geq \totalmw_{\viweight_t}(S')\right\} 
     %=\mathbb{I}\left\{\sum_{e \in S} \ucbI_{e,t} \geq \sum_{e \in S'}\ucbI_{e,t}\right\}
        \textstyle = \mathbb{I}\left\{ \sum_{i=1}^A \iweight_{e_i,t} \geq \sum_{j=1}^B \iweight_{e'_j,t} \right\} \\
        &\leq \mathbb{I} \left[  \scriptstyle 
         \max\limits_{l \leq c_1,\dots,c_A <t}  \sum\limits_{i=1}^A \left(  \sweight_{e_i,c_i}+\sqrt{\frac{(\noLink+1)\ln t}{c_i}} \right) \right. \\
            &\hspace{2cm}\scriptstyle ~~~~~~ \left.\geq \min\limits_{0 < c'_1,\dots,c'_B<t} \sum\limits_{j=1}^B  \left(  \sweight_{e'_j,c'_j}+\sqrt{\frac{(\noLink+1)\ln t}{c'_j}} \right) \right], \\
\end{split}
\end{equation}
which is bounded by 
$
 \textstyle \sum^{t}_{c_1=1} \cdots \sum^{t}_{c_A=1} \sum^{t}_{c'_1=1} \cdots \sum^{t}_{c'_B=1}
 \mathbb{I} \{\mathbb{Z}\},
$
where event $\mathbb{Z}$ is defined as 
\[\scriptstyle 
\left\{ \sum\limits_{i=1}^A \left( \sweight_{e_i,c_i}+\sqrt{\frac{(\noLink +1)\ln t}{c_i}} \right) \geq \sum\limits_{j=1}^B \left( \sweight_{e'_j,c'_j}+\sqrt{\frac{(\noLink +1)\ln t}{c'_j}}  \right)\right\}.
\]
We further decompose $\mathbb{Z}$ into $\{\mathbb{A}_i\}$, $\{\mathbb{B}_j\}$ and $\mathbb{C}$ as
\begin{equation}
\begin{split}
    \mathbb{A}_i :&\textstyle  \sweight_{e_i,c_i}-\sqrt{\frac{(\noLink +1)\ln t}{c_i}} \geq \mweight_{e_i},\\
    \mathbb{B}_j :&\textstyle  \sweight_{e_j,c'_j}+\sqrt{\frac{(\noLink +1)\ln t}{c'_j}} \leq \mweight_{e'_j},\\
    \mathbb{C}   :&\textstyle  \sum_{j=1}^B \mweight_{e'_j} < \sum_{i=1}^A ( \mweight_{e_i} +2\sqrt{\frac{(\noLink +1)\ln t}{c_i}} ).
\end{split}
\end{equation}
Note that if $\mathbb{I}\{\mathbb{Z}\}=1$, we should have $\sum_{i=1}^A \mathbb{I}\{\mathbb{A}_i\} + \sum_{j=1}^B \mathbb{I}\{\mathbb{B}_j\} + \mathbb{I}\{\mathbb{C}\} \ge 1$, which can be easily shown by contradiction. The probability of each event of $\{\mathbb{A}_i\}$ and $\{\mathbb{B}_j\}$ can be bounded by the Chernoff-Hoeffding bound as
\begin{equation}
\begin{split}
\textstyle \Pr \{ \sweight_{e_i,c_i}-\sqrt{\frac{(\noLink +1)\ln t}{c_i}} \geq \mweight_{e_i} \}
    &\leq t^{-2(\noLink +1)},\\
\textstyle \Pr \{ \sweight_{e'_j,c'_j}+\sqrt{\frac{(\noLink +1)\ln t}{c'_j}} \leq \mweight_{e'_j} \} 
    &\leq t^{-2(\noLink +1)}. 
\end{split}
\end{equation}
Further, we can show $\Pr\{\mathbb{C}\}=0$ if $ c_i \geq \lceil \frac{4\noLink ^2(\noLink +1)\ln t}{(\alpgapmin)^2} \rceil$ for all links $e_i \in S$ as follows. Suppose that $\mathbb{C}$ occurs. Then we should have $0 > \sum_{j=1}^B \mweight_{e'_j} - \sum_{i=1}^A \mweight_{e_i} - \sum_{i=1}^A 2\sqrt{\frac{(\noLink+1)\ln t}{c_i}}$. The right term is no smaller than $\alpha\cdot \maxtotalmw_{\vmweight} -\max_{S \in \cnopts_{\vmweight}}r_{\boldsymbol{\mweight}}(S)  -\alpgapmin$, which equals $0$ by definition of $\alpgapmin$, resulting in a contradiction.

Using these and taking expectation on (\ref{eq:event_reward}), we have
\begin{equation*}
\begin{split}
&\Pr \textstyle \left\{ r_{ \boldsymbol{ \iweight_t} }(S) \geq r_{ \boldsymbol{ \iweight_t } } (S') \right\}
    \leq \sum\limits^{t}_{c_1=1} \cdots \sum\limits^{t}_{c_A=1}\sum\limits^{t}_{c'_1=1} \cdots \sum\limits^{t}_{c'_B=1} 
        \Pr\{\mathbb{Z} \}\\
    &\textstyle \leq \sum\limits^{t}_{c_1=1} \cdots \sum\limits^{t}_{c_A=1}\sum\limits^{t}_{c'_1=1} \cdots \sum\limits^{t}_{c'_B=1}   
        \left( \sum\limits_{i=1}^A\Pr\{\mathbb{A}_i\} + \sum\limits_{j=1}^B\Pr\{\mathbb{B}_j\} \right) \\
    %&\leq \sum^{t}_{c_1=1}\cdot\cdot\sum^{t}_{c_A=1}\sum^{t}_{c'_1=1}\cdot\cdot\sum^{t}_{c'_B=1}   2\noLink t^{-2(\noLink +1)}\\
    &\leq t^{A+B} \cdot 2\noLink t^{-2(\noLink+1)}    \leq 2\noLink t^{-2}.
\end{split}
\end{equation*}
%}{}
%%%%%%%%%%%%%% end of app:C
}

\iftoggle{tech_report}
{
%%%%%%%%%%%%%%%%%%%%%%%%%%%%
% app:D
\section{Proof of Proposition~\ref{pro:regretbound}} \label{app:D}

Overall, we show that the number of explorations to non-near-optimal matchings is bounded. To this end, we consider a sequence of time points where a non-near-optimal matching is sufficiently played at each point. They serve as a foothold to count the total number of plays of non-near-optimal matchings.

To begin with, for an arbitrary fixed time $h>0$, let $l_h = \lceil \frac{4\noLink^2(\noLink+1)\ln h}{(\alpgapmin)^2} \rceil$, and let $\hat{t}_h$ denote the first time when all non-near-optimal matchings are sufficiently (i.e., more than $l_h$ times) explored, i.e., 
\begin{align*}
\hat{t}_h = \min \left\{ t ~\middle|~ 
    \nochosen_S(t) \geq  l_h ~\text{for all}~ S \in \cnopts_{\vmweight} \right\}. 
\end{align*}

\noindent \textbf{(1) When} $\hat{t}_h \leq h :$ 
Let $\cnopts_{\vmweight} = \{S^1,S^2,...,S^M\}$ with $M = \vert \cnopts_{\vmweight} \vert$. Further we define $\overline{\setS}(t) = \{ S \in \cnopts_{\vmweight} ~|~ \nochosen_S(t) \geq l_h\}$, which is the set of non-near-optimal matchings that are scheduled sufficiently many times by time $t$, and $\underline{\setS}(t) = \cnopts_{\vmweight} - \overline{\setS}(t)$ denotes the set of not-yet-sufficiently-scheduled non-near-optimal matchings.  
Also, let $t^n$ denote the time when matching $S^n$ is sufficiently scheduled, i.e., $\nochosen_{S^n}(t^n) = l_h$. Without loss of generality, we assume $t^1 < t^2 < \dots < t^M = \hat{t}_h$.

To apply the decomposition inequality (\ref{eq:decomposition}), we need to estimate the expected value of $\sum_{S \in \cnopts_{\vmweight}} \nochosen_S(\hat{t}_h)$, which can be written as
\begin{equation} \label{eq:tau_samplepath}
\begin{split}
&\textstyle \sum_{S \in \cnopts_{\vmweight}} \nochosen_S(\hat{t}_h) 
    \textstyle= \sum_{S \in \cnopts_{\vmweight}} \sum_{t=1}^{\hat{t}_h} \mathbb{I}\{\schedule_t=S\} \\
    %&= \sum_{j=1}^M \sum_{t=1}^{\hat{t}_h} \mathbb{I}\{\schedule_t=S^j\} \\
    &\textstyle = l_h M + \sum_{n=1}^{M-1}\sum_{t=t^n+1}^{t^{n+1}} \sum_{S \in \overline{\setS}(t^n)} \mathbb{I}\{ \schedule_t =S\}.
    %&= l_h M + \sum_{n=1}^{M-1} \sum_{t=t^n+1}^{t^{n+1}} \sum_{j=1}^n  \mathbb{I}\{ \schedule_t =S^j \}.
\end{split}
\end{equation}
%
%\iftoggle{tech_report}
%{
    Hence, we need to estimate $\sum_{S \in \overline{\setS}(t^n)} \Pr \{ \schedule_t =S \}$ for $t \in (t^n, t^{n+1}]$, which can be obtained as in the following lemma.
%}

\begin{lemma} \label{lem:prob-for-sufficiently-played-matching}
For each $t \in (t^n, t^{n+1}]$, we have 
    \begin{equation} \label{eq:upper-bound-for-sufficiently-played-matching}
    \begin{split}
        &\textstyle \sum_{S \in \overline{\setS}(t^n)} \Pr \{ \schedule_t =S \}\\
            &\textstyle \leq (1-\delta) \cdot \sum_{S \in \overline{\setS} (t^n)} \Pr \{\schedule_{t-1}=S\} \\
            &\textstyle + \Pr \{\schedule_{t-1} \in \underline{\setS}(t^n) \} +\left( \vert \overline{\setS}(t^n) \vert + \delta \right) \cdot 2\noLink t^{-2}.
    \end{split}
    \end{equation}
\end{lemma}
\begin{proof}
We first divide the case into three exclusive sub-cases based on the previous schedule $\schedule_{t-1}$: 
%From $\overline{\setS}(t^n) \cup \underline{\setS}(t^n) \cup \nopts_{\vmweight} = \setS$, we define three 
events $\mathbb{A}=\{ \schedule_{t-1} \in \nopts_{\vmweight} \}$, $\mathbb{B}=\{\schedule_{t-1} \in \underline{\setS}(t^n)\}$, and $\mathbb{C}=\{\schedule_{t-1} \in \overline{\setS}(t^n)\}$. Then we have 
\begin{align}
&\textstyle \sum_{S \in \overline{\setS}(t^n)} \Pr \{ \schedule_t =S \} \nonumber \\
%&&\sum_{j=1}^n \Pr \{ \schedule_t = S^j \} \nonumber \\
    &= \textstyle \sum_{S \in \overline{\setS}(t^n)} \Pr \{ \schedule_t =S ~\vert~ \mathbb{A} \} \cdot \Pr \{ \mathbb{A} \} \label{eq:first}\\
    %&&~~= \sum_{j=1}^n \Pr \{ \schedule_t = S^j \vert \schedule_{t-1} \in \nopts_{\vmweight} \} \cdot \Pr (\schedule_{t-1} \in \nopts_{\vmweight}) \label{eq:first}\\
    &~~\textstyle + \sum_{S \in \overline{\setS} (t^n)} \Pr \{ \schedule_t =S ~\vert~ \mathbb{B}\} \cdot \Pr \{\mathbb{B}\} \label{eq:second} \\
    %&&~~~~+ \sum_{j=1}^n \Pr \{ \schedule_t = S^j \vert \schedule_{t-1} \in \underline{\setS}(t^n)\} \cdot \Pr \{\schedule_{t-1} \in \underline{\setS}(t^n)\} \label{eq:second} \\
    &~~\textstyle + \sum_{S \in \overline{\setS} (t^n)} \Pr \{ \schedule_t =S ~\vert~ \mathbb{C}\} \cdot \Pr \{\mathbb{C}\}. \label{eq:third}
    %&&~~~~+ \sum_{j=1}^n \Pr \{ \schedule_t = S^j \vert \schedule_{t-1} \in \overline{\setS}(t^n)\} \cdot \Pr \{\schedule_{t-1} \in \overline{\setS}(t^n)\}. \label{eq:third}
\end{align}
Let $\mathcal{A}_t$ denote the set of augmentations chosen under our algorithm at time $t$. We can obtain a bound on (\ref{eq:first}) as
\begin{align}
    &\textstyle\sum_{S \in \overline{\setS}(t^n)} \Pr \{ \schedule_t =S ~\vert~ \mathbb{A} \} \cdot \Pr \{\mathbb{A} \} \nonumber\\
    %&~~\textstyle\leq \sum_{S \in \overline{\setS}(t^n)} \Pr \{\schedule_{t-1} \in \nopts_{\vmweight}\} \cdot \Pr \{ S = \schedule_{t-1} \oplus \mathcal{A}_t \} \nonumber \\
    %&~~\textstyle\leq \sum_{S \in \overline{\setS}(t^n)} \Pr \{ \mathbb{A} \} \cdot \Pr \{ S = (\schedule_{t-1} \oplus \mathcal{A}_t) \} \nonumber \\
        %&~~~~~~~~~~~~ \textstyle \cdot \Pr \{ \totalmw_{\viweight_t}(S)  \geq \totalmw_{\viweight_t}(\schedule_{t-1}) ~|~ \schedule_{t-1} \in \nopts_{\vmweight} \} \nonumber \\
        %&~~~~~~~~~~~~~~~~ \textstyle \cdot \Pr \{ \totalmw_{\viweight_t}(S)  \geq \totalmw_{\viweight_t}(\schedule_{t-1}) ~|~ \mathbb{A} \} \nonumber \\
    %&~~\textstyle \leq \sum_{S \in \overline{\setS}(t^n)} \Pr \{\totalmw_{\viweight_t}(S)  \geq \totalmw_{\viweight_t}(\schedule_{t-1})~|~\schedule_{t-1} \in \nopts_{\vmweight} \} \nonumber\\
    &~~\textstyle \leq \sum_{S \in \overline{\setS}(t^n)} \Pr \{\totalmw_{\viweight_t}(S)  \geq \totalmw_{\viweight_t}(\schedule_{t-1})~|~\mathbb{A}\} \cdot \Pr \{ \mathbb{A} \} \nonumber\\
    &~~\leq \vert \overline{\setS}(t^n) \vert \cdot 2 \noLink t^{-2}, \label{eq:4}
\end{align}
where the last inequality comes from Lemma \ref{lem:nearoptimal}. The result holds for all $t \in (t^n,t^{n+1}]$. 
For the second term (\ref{eq:second}), we have
\begin{align}
\textstyle \sum_{S \in \overline{\setS} (t^n)} 
    %& \Pr \{ \schedule_t =S \vert \schedule_{t-1} \in \underline{\setS}(t^n)\} \cdot \Pr \{\schedule_{t-1} \in \underline{\setS}(t^n)\}\nonumber \\ 
    \Pr \{ \schedule_t =S ~\vert~ \mathbb{B} \} \cdot \Pr \{\mathbb{B} \}
    %&\leq \Pr \{\schedule_{t-1} \in \underline{\setS}(t^n) \}. \label{eq:5}
    \leq \Pr \{\mathbb{B}  \}. \label{eq:5}
\end{align}
Finally, the third term (\ref{eq:third}) denotes the probability to transit from a sufficiently-played non-near-optimal matching to a sufficiently-played non-near-optimal matching, and thus we have
\begin{align*}
    &\textstyle \sum_{S \in \overline{\setS} (t^n)} \Pr \{ \schedule_t = S ~\vert~ \mathbb{C} \} \cdot 
    \Pr \{ \mathbb{C} \} \\
        &\textstyle = \sum_{S \in \overline{\setS} (t^n)} \Pr \{ \schedule_t \in \overline{\setS}(t^n) \vert \schedule_{t-1}=S \} \cdot \Pr \{\schedule_{t-1}=S\}.
\end{align*}
%Then, for $\alpha = \alpk$ and $S \in \overline{\setS}(t^n)$, 
Letting $S' = S \oplus \mathcal{A}_t$ and using Lemma~\ref{lem:probability}, the conditional probability can be derived as
\[
\begin{split}
&\Pr \{ \schedule_t \in \overline{\setS}(t^n) \vert \schedule_{t-1}=S \}\\
    %&\le \Pr \{ \schedule_t \in \overline{\setS}(t^n) \vert \schedule_{t-1}=S, S' \in \nopts_{\vmweight} \} \cdot \Pr \{ S' \in \nopts_{\vmweight}\} \\
    %\\&~~~~~~+ 
    %\Pr \{ \schedule_t \in \overline{\setS}(t^n) \vert \schedule_{t-1}=S, S \oplus \mathcal{A}_t \in \cnopts_{\vmweight} \} \cdot 
    %&\hspace{6cm} + \Pr \{S' \in \cnopts_{\vmweight}\}\\
    &\leq \Pr \{ \schedule_t \in \overline{\setS}(t^n) \vert \schedule_{t-1}=S, S' \in \nopts_{\vmweight} \} \cdot \delta + (1-\delta)\\
    &= \Pr \{\totalmw_{\viweight_t}(S) \geq \totalmw_{\viweight_t}(S')\}\cdot \delta + (1-\delta)\\
    %&~= \Pr \left\{\sum_{e \in S} \ucbI_{e,t} \geq \sum_{e \in C_t}\ucbI_{e,t}\middle|~S \in \overline{\setS}(t^n),~ C_t\in \nopts_{\vmweight}\right\}\cdot \delta +1-\delta\\
    &\leq 2\noLink t^{-2}\cdot \delta + (1-\delta). 
\end{split}
\]
where %the second inequality comes from Lemma~\ref{lem:probability}, 
the equality holds since $\schedule_t$ should be $S$ (otherwise, $\schedule_t = (S \oplus \mathcal{A}_t) \notin \overline{\setS}(t^n)$) and thus $S$ should have the larger weight sum to be chosen by the augmentation algorithm, and the last inequality comes from Lemma~\ref{lem:nearoptimal}.
%
%\iftoggle{tech_report}
%{
    Hence, the third term (\ref{eq:third}) can be upper bounded by
    \begin{equation}  \label{eq:6} 
    \begin{split}
        &\textstyle \sum_{S \in \overline{\setS} (t^n)} \Pr \{\schedule_{t-1}=S\} 
            \cdot \left( 2 \delta \noLink t^{-2} + 1 - \delta \right) \\
        &\textstyle \le 2 \delta \noLink t^{-2} 
            + \left( 1 - \delta \right) \sum_{S \in \overline{\setS} (t^n)} \Pr \{\schedule_{t-1}=S\},
    \end{split}
    \end{equation}
    for all $t \in (t^n,t^{n+1}]$. 

    The result can be obtained by combining (\ref{eq:4}), (\ref{eq:5}), and (\ref{eq:6}).
%}
%{
%    Combining these results, we can obtain (\ref{eq:upper-bound-for-sufficiently-played-matching}).
%}
\end{proof}

In order to apply Lemma~\ref{lem:prob-for-sufficiently-played-matching} to (\ref{eq:tau_samplepath}), we rewrite it in a recursive form. Let $\eta = 1 - \delta$, $G_n = \left( \vert \overline{\setS}(t^n) \vert +\delta \right) \cdot 2\noLink = (n + \delta) \cdot 2\noLink$, and $\Theta_n(t) = \sum_{S \in \overline{\setS}(t^n)} \Pr \{ \schedule_t =S \}$. We have a recursive form of  (\ref{eq:upper-bound-for-sufficiently-played-matching}) as
\[
\Theta_n(t) 
    \leq  \Pr \{\schedule_{t-1} \in \underline{\setS}(t^n) \} + G_n \cdot t^{-2} + \eta \Theta_n(t-1),
\]
for $t \in (t^n, t^{n+1}]$. 
%\iftoggle{tech_report}
%{
    Extending the right side further down to $t^n$, we can obtain that
    \begin{align} 
    \Theta_n(t) 
        &\le \eta^{t-t^n} \Theta_n(t^n) \label{eq:7} \\
        &\textstyle~~+ G_n \sum_{i=t^n+1}^{t} \eta^{t-i} \cdot i^{-2} \label{eq:8} \\
        %&~~+ G_n (t^{-2} +\eta (t-1)^{-2} \cdots +\eta^{t-t^n-1} (t^n+1)^{-2} ) \\
        &\textstyle~~+ \sum_{i=t^n+1}^{t} \eta^{t-i} \cdot \Pr \{\schedule_{i-1} \in \underline{\setS}(t^n) \}. \label{eq:9}
    \end{align}
    By summing it over $t \in (t^n, t^{n+1}]$ on the both sides, we obtain the following lemma.
%}
%{
%    Extending the right side further down to $t^n$ and summing over $t \in (t^n, t^{n+1}]$, we obtain the following lemma.
%}
%
\begin{lemma} \label{lem:theta}
The total number of times that sufficiently played non-near-optimal matchings are selected during $(t^n,t^{n+1}]$ is bounded by 
\begin{equation}\label{eq:reuse}
\textstyle \sum_{t=t^n+1}^{t^{n+1}} \Theta_n(t)
    \leq \frac{1}{\delta}
    ( 
        1+ \frac{\pi^2}{6} G_n  + \mathbb{E} [ \sum_{S \in \underline{\setS}(t^n)} \tau_{S,n+1} ]
    ),
\end{equation}
where $\tau_{S,n+1}$ denote the number of time slots that $S$ is scheduled in $(t^{n},t^{n+1}]$.
\end{lemma}
\iftoggle{tech_report}
{
    \noindent The proof of Lemma~\ref{lem:theta} follows the same line of analysis as~\cite{kang_tmc20} and included for the completion.
    
    \begin{proof}
    By summing up (\ref{eq:9}) over $t \in(t^n,t^{n+1}],$ we have
    \begin{align}
        &\hspace{-0.2cm}\textstyle \sum_{t=t^n +1}^{t^{n+1}} \sum_{i=t^n+1}^{t} \eta^{t-i} \cdot \Pr \{ \schedule_{i-1} \in \underline{\setS}(t^n) \} \nonumber\\
        &\textstyle =\sum_{t=t^n+1}^{t^{n+1}} \Pr \{ \schedule_{t-1} \in \underline{\setS}(t^n) \}  \cdot (\sum^{t^{n+1} -t}_{i=0} \eta^i)  \nonumber\\
        %&\textstyle =\sum_{t=t^n }^{t^{n+1}-1}\sum^{t^{n+1} -t-1}_{s=0} \beta^s \cdot \mathbb{E}\left[ \mathbb{I}\{ \schedule_{t} \in \underline{\setS}(t^n) \} \right] \nonumber\\
        %&\textstyle \leq\frac{1}{1-\beta}\mathbb{E}\left[ \sum^{t^{n+1}-1}_{t=t^n} \mathbb{I}\{  \schedule_{t} \in \underline{\setS}(t^n) \}\right]
        %\stackrel{(a)}{=}\frac{1}{1-\beta}\mathbb{E}\left[ \sum^{t^{n+1}-1}_{t=t^n+1} \mathbb{I}\{  \schedule_{t} \in \underline{\setS}(t^n) \}\right] \nonumber\\
        &\textstyle \stackrel{(a)}{\leq} \frac{1}{1-\eta} \mathbb{E}\left[ \sum^{t^{n+1}}_{t=t^n+1} \mathbb{I}\{ \schedule_{t} \in \underline{\setS}(t^n) \}\right]\nonumber\\
        &\textstyle =\frac{1}{\delta} \mathbb{E}\left[ \sum_{S \in \underline{\setS}(t^n)} \tau_{S,n+1}\right], \label{eq:lem-theta-1} 
    \end{align}
    where $(a)$ holds because $\sum^{t^{n+1} -t}_{i=0} \eta^i \le \frac{1}{1-\eta}$, $\Pr \{ \schedule_{t-1} \in \underline{\setS}(t^n) = \mathbb{E}\left[ \mathbb{I}\{ \schedule_{t-1} \in \underline{\setS}(t^n) \} \right]$, and $\mathbb{I}\{ \schedule_{t^n}\in \underline{\setS}(t^n)\}=0$ since the matching scheduled at $t^n$ does not belong to $\underline{\setS}(t^n)$ by definition.
    
    Also, by summing up (\ref{eq:8}) over $(t^n,t^{n+1}]$, we have
    \begin{equation} \label{eq:lem-theta-2} 
    \textstyle 
    \sum_{t=t^n+1}^{t^{n+1}} G_n \sum_{i=t^n+1}^{t} \eta^{t-i} \cdot i^{-2} 
        \le G_n \cdot \frac{1}{\delta}\cdot \frac{\pi^2}{6}.
    \end{equation}
    Similarly, we take the sum of (\ref{eq:7}) as
    \begin{equation} \label{eq:lem-theta-3}
    \textstyle
    \sum_{t=t^n+1}^{t^{n+1}}\eta^{t-t^n} \Theta_n(t^n)
        \le \sum_{t=t^n+1}^{t^{n+1}} \eta^{t-t^n}
        \le \frac{1}{\delta}.
    \end{equation}
    Combining (\ref{eq:lem-theta-1}), (\ref{eq:lem-theta-2}), and (\ref{eq:lem-theta-3}), we obtain the result.
    \end{proof}
}
{
    \noindent The proof of Lemma~\ref{lem:theta} is omitted and can be found in~\cite{tr}.
}

Now, by taking expectation on (\ref{eq:tau_samplepath}), we can obtain the expected total number of times that non-near optimal matchings are selected up to time $\hat{t}_h (\leq h)$  as 
\begin{align*}
    &\textstyle \sum_{S\in \cnopts_{\vmweight}} \mathbb{E}\left[ \nochosen_S(\hat{t}_h) \right] 
        %=\sum_{S\in \cnopts_{\vmweight}}\sum_{t=1}^{t'} \Pr \{ \schedule_t =S \}\\
        \textstyle = l_h M + \sum_{n=1}^{M-1} \sum_{t=t^n + 1}^{t^{n+1}} \Theta_n(t)\\
        &\textstyle \leq l_hM + \frac{1}{\delta} \sum_{n=1}^{M-1} \left( 1+ G_n \frac{\pi^2}{6} + \mathbb{E}\Big[ \sum_{S \in \underline{\setS}(t^n)} \tau_{S,n+1} \Big] \right),\\
        %&\textstyle \leq l_hM + \left(\frac{1}{\delta} + \frac{\noLink\pi^2}{3}\right) (M-1)
        %+\frac{\noLink\pi^2}{3\delta}\cdot\sum_{n=1}^{M-1} \vert \overline{\setS}(t^n)\vert
        %+\frac{1}{\delta}\cdot \sum_{n=1}^{M-1} \mathbb{E}\left[ \sum_{s=n+1}^{M} \tau_{s,n+1}\right]\\
        %&\textstyle \leq l_hM + (M-1)\left(\frac{1}{\delta} + \frac{\noLink\pi^2}{3}\right) +\frac{\noLink\pi^2}{3\delta}\cdot \frac{M(M-1)}{2} +\frac{1}{\delta}\cdot l_hM\\
        &\textstyle \le l_hM + \frac{M}{\delta} \left( 1 + \frac{\noLink \delta \pi^2}{3} +\frac{\noLink(M-1)\pi^2}{6} +l_h \right).
\end{align*}
%
%\iftoggle{tech_report}
%{
    The last inequality holds since  
    (i) $\sum_{n=1}^{M-1} G_n = \sum_{n=1}^{M-1} (n + \delta) \cdot 2\noLink 
    \le M \cdot \noLink \cdot (2\delta + (M-1))$, and 
    (ii) $\sum_{S \in \underline{\setS}(t^n)} \tau_{S,n+1}$ is the total number that the matchings that have been chosen less than $l_h$ up to $t^n$ are chosen during $(t^n, t^{n+1}]$ and thus results in $\sum_{n=1}^{M-1} \sum_{S \in \underline{\setS}(t^n)} \tau_{S,n+1} \leq \sum_{k=2}^{M} l_h \le l_hM$.
%}{}
%
From $M\leq\vert \setS \vert -1,$ we have 
\begin{align} 
\textstyle \sum_{S\in \cnopts_{\vmweight}} \mathbb{E}\left[ \nochosen_S(\hat{t}_h) \right]
    %\textstyle &\leq \left( 1+\frac{1}{\delta} \right)l_hM +D_1 \\
    %&\leq \left( 1+\frac{1}{\delta} \right) \cdot  \frac{4\noLink^2(\noLink+1)\ln h}{(\alpgapmin)^2} ( \vert \setS \vert -1 ) +D_1,
    &\textstyle \leq D_1 \cdot \frac{\ln h}{(\alpgapmin)^2} \nonumber\\
        &\textstyle \hspace{-1cm} + \frac{\vert\setS\vert-1}{\delta} (1+ \frac{\noLink\delta\pi^2}{3} + \frac{\noLink(\vert\setS\vert-2)\pi^2}{6} ). \label{eq:16}
\end{align}

This provides a bound on the number of times that non-near-optimal matchings are selected up to $\hat{t}_h$. For the rest time $t \in (\hat{t}_h,h]$, we need to compute $\sum_{t=\hat{t}_h+1}^h \Pr \{\schedule_t \in \cnopts_{\vmweight}\}$. Let $S' = \schedule_{t-1} \oplus \mathcal{A}_t$. Since next schedule $\schedule_t$ is either $\schedule_{t-1}$ and $S'$ under the algorithm, we divide the event $\{\schedule_t \in \cnopts_{\vmweight}\}$ into three sub-cases based on $\schedule_{t-1}$ and $S'$, and compute the probability as 
\begin{align}
%&\textstyle \sum_{t=\hat{t}_h+1}^h \sum_{S \in \cnopts_{\vmweight}}  \Pr \{ \schedule_t = S \} \\
\textstyle \Pr \{ \schedule_t \in \cnopts_{\vmweight}\} 
    &\textstyle =  \Pr\{S' \in\cnopts_{\vmweight},\schedule_{t-1}\in\cnopts_{\vmweight}\}\nonumber \\
    &\textstyle \hspace{-0.5cm}+ \Pr\{S' \in \cnopts_{\vmweight},~\schedule_{t-1}\in\nopts_{\vmweight}, \totalmw_{\vmweight}(S')\geq\totalmw_{\vmweight}(\schedule_{t-1})\}\nonumber\\
    &\textstyle \hspace{-0.5cm}+ \Pr\{S' \in \nopts_{\vmweight},~\schedule_{t-1}\in\cnopts_{\vmweight}, \totalmw_{\vmweight}(S')\leq\totalmw_{\vmweight}(\schedule_{t-1})\}.\nonumber
\end{align}
This leads to
\begin{align}
\textstyle \Pr \{ \schedule_t \in \cnopts_{\vmweight}\} 
    &\textstyle \leq  \Pr\{S' \in \cnopts_{\vmweight}~|~\schedule_{t-1} \in\cnopts_{\vmweight}\} \cdot \Pr\{\schedule_{t-1}\in\cnopts_{\vmweight} \} \nonumber \\
    &\textstyle \hspace{-0.5cm}+ \Pr\{\totalmw_{\vmweight}(S')\geq\totalmw_{\vmweight}(\schedule_{t-1}) ~|~S' \in \cnopts_{\vmweight},~\schedule_{t-1}\in\nopts_{\vmweight}\} \nonumber\\
    &\textstyle \hspace{-0.5cm}+  \Pr\{\totalmw_{\vmweight}(S')\leq\totalmw_{\vmweight}(\schedule_{t-1})~|~S' \in \nopts_{\vmweight},~\schedule_{t-1}\in\cnopts_{\vmweight}\}.\nonumber
\end{align}
From Lemma~\ref{lem:probability}, we have $\Pr\{S' \in \cnopts_{\vmweight}~|~\schedule_{t-1} \in\cnopts_{\vmweight}\} = 1 - \Pr\{S' \in\nopts_{\vmweight}~|~\schedule_{t-1} \in \cnopts_{\vmweight}\} \le 1-\delta = \eta$. Since $\hat{\tau}_S(t) \geq l_h$ for all $S$ and $t \in (\hat{t}_h,h]$, Lemma~\ref{lem:nearoptimal} provides an upper bound $2\noLink t^{-2}$ on each conditional probability in the second and the third terms. As a result, we can obtain
\[ \textstyle 
    \Pr \{\schedule_t\in \cnopts_{\vmweight}\} 
        \textstyle \leq  \eta \cdot \Pr\{\schedule_{t-1} \in \cnopts_{\vmweight}\} + 4\noLink t^{-2}.
\]
By extending the inequality in a recursive manner down to $\hat{t}_h$, we obtain that 
\[
\begin{split}
    \Pr \{\schedule_t \in \cnopts_{\vmweight}\} 
        &\textstyle \leq \eta^{t-\hat{t}_h} \cdot \Pr \{ \schedule_{\hat{t}_h}\in \cnopts_{\vmweight} \} 
        + 4\noLink \sum_{i=\hat{t}_h+1}^t \eta^{t-i}i^{-2} \\
        &\textstyle = \eta^{t-\hat{t}_h}  
        + 4\noLink \sum_{i=\hat{t}_h+1}^t \eta^{t-i}i^{-2},
\end{split}
\]
where the last equality holds since $\Pr \{ \schedule_{\hat{t}_h}\in \cnopts_{\vmweight} \} = 1$ from the definition of $\hat{t}_h$. Summing over $t \in (\hat{t}_h, h]$ on the both sides, and from $\eta = 1-\delta$, we have
\begin{equation} \label{eq:19} 
    \textstyle
    \sum_{t=\hat{t}_h+1}^h %\sum\limits_{S \in \cnopts_{\vmweight}} \Pr \{\schedule_t=S\}
    \Pr \{\schedule_t\in \cnopts_{\vmweight}\}
        \le \frac{1-\delta}{\delta}+ \frac{2\noLink\pi^2}{3\delta}.
\end{equation}

Combining (\ref{eq:16}) and (\ref{eq:19}), we obtain
\begin{equation} \label{eq:prop2_key1}
\begin{split}
    &\textstyle \sum_{S\in \cnopts_{\vmweight}} \mathbb{E} [ \nochosen_S(h) ]\\
        &=\textstyle \sum_{S\in \cnopts_{\vmweight}} \mathbb{E}\left[ \nochosen_S(\hat{t}_h) \right]
            + \sum_{t=\hat{t}_h+1}^h \Pr \{\schedule_t\in \cnopts_{\vmweight}\} \\
        &\leq \textstyle  D_1 \cdot \frac{\ln h}{(\alpgapmin)^2} + D_2, 
\end{split}
\end{equation}
where 
%new D_1 and D_2
$D_1 = \left( 1 + \frac{1}{\delta} \right) \cdot 4\noLink^2(\noLink+1) \cdot \left( \vert \setS \vert-1 \right)$, and
$D_2 = \frac{\vert\setS\vert-1}{\delta} (1+ \frac{\noLink\delta\pi^2}{3} + \frac{\noLink(\vert\setS\vert-2)\pi^2}{6} ) + \frac{1-\delta}{\delta} + \frac{2\noLink\pi^2}{3\delta}$.
% old D_1 and D_2
%$D_1 = \frac{\vert\setS\vert-1}{\delta}\left(1+ \frac{\noLink\pi^2\delta}{3}+ \frac{\noLink(\vert\setS\vert-2)\pi^2}{6} \right)$, 
%$D_2=\frac{1-\delta}{\delta}+ \frac{2\noLink\pi^2}{3\delta}$, 
%and $\delta = \min \{ 1, ( \frac{p}{1-p} )^V \} ( \frac{1-p}{k\Sigma} )^V$.

\bigskip

%\iftoggle{tech_report}{}
\textbf{(2) When} $\hat{t}_h > h$ (i.e,. $\exists S$ such that $\nochosen_S(h) < l_h$) : With the same definitions of $l_h$, $\overline{\setS}(t)$, and $\underline{\setS}(t)$, let $\vert \overline{\setS} \vert = |\overline{\setS}(h)|$ and $\vert \underline{\setS} \vert = |\overline{\setS}(h)|$. At this time, we define $\snopts(h)=\{ S^1,S^2,...,S^{\vert\overline{\setS}\vert} \}$ and let $t^n$ denote the time at which matching $S^n$ is sufficiently scheduled, i.e., $\nochosen_{S^n}(t^n)=l_h$. Without loss of generality, we assume $t^1 < t^2<...<t^{\vert \overline{\setS}\vert}$. By time slot $h$, $\underline{\setS}(h)$ is non-empty (since $\hat{t}_h >h)$, and it is clear that $\sum_{S \in \underline{\setS}(h)} \nochosen_S(h) \leq l_h \vert \underline{\setS} \vert$. 

Similar to the case when $\hat{t}_h \le h$, we can obtain 
\[
\begin{split}
&\textstyle \sum_{S \in \cnopts_{\vmweight}} \mathbb{E}\left[ \nochosen_S(h) \right] \\
    &\textstyle =\sum_{S \in \underline{\setS}(h)} \mathbb{E} [ \nochosen_S(h) ]
        %+\sum_{S \in \overline{\setS}(h)} \mathbb{E} [ \nochosen_S(h) ] \\
        + \sum_{t=1}^h\sum_{S \in \overline{S}(h)}\Pr \{ \schedule_t =S \}\\
    %&\textstyle \leq l_h\vert\underline{\setS} \vert + \sum_{t=1}^h\sum_{S \in \overline{S}(h)}\Pr \{ S_t =S \} \\
    &\textstyle \leq l_h\vert\underline{\setS} \vert + l_h\vert \overline{\setS}\vert+ \sum_{n=1}^{\vert \overline{\setS} \vert}\sum_{t=t^n+1}^{t^{n+1}} \Theta_n(t) \\
    &\textstyle \leq l_h M + \sum_{n=1}^{\vert \overline{\setS} \vert} \frac{1}{\delta} 
        %\left( 1 + \frac{\pi^2}{6} G_n + \mathbb{E}\left[\sum_{x=n+1}^{\vert\overline{\setS}\vert} \tau_{x,n+1} \right] \right),
        ( 1 + \frac{\pi^2}{6} G_n + \mathbb{E} [ \sum\limits_{S \in \underline{\setS}(t^n)} \tau_{x,n+1} ] ),
\end{split}
\]
where the last inequality comes from Lemma~\ref{lem:theta}. As in (\ref{eq:16}), we can obtain 
\begin{equation}\label{eq:prop2_key2}
    \begin{split}
    &\textstyle \sum_{S \in \cnopts_{\vmweight}} \mathbb{E}\left[ \nochosen_S(h) \right]\\
    &\le \textstyle l_h M + \frac{\vert \overline{\setS} \vert}{\delta} \left(1 + \frac{\noLink \pi^2 \delta}{3} + \frac{ \noLink (\vert \overline{\setS} \vert + 1) \pi^2}{6} + l_h \right)\\
    %&\leq l_h(\vert \setS\vert-1) + \vert\overline{\setS}\vert\left(\frac{1}{\delta}+\frac{\noLink\pi^2}{3}\right) + \frac{\noLink\pi^2}{6\delta}\vert \overline{\setS}\vert(\vert \overline{\setS}\vert+1) +l_h\frac{\vert \overline{\setS}\vert}{\delta}, 
    %&\textstyle \leq \left( 1+\frac{1}{\delta} \right) \cdot  \frac{4\noLink^2(\noLink+1)\ln h}{(\alpgapmin)^2} ( \vert \setS \vert -1 ) +D_1 +D_2,
    &\leq \textstyle D_1 \frac{\ln h}{(\alpgapmin)^2} + D_2, 
    \end{split}
\end{equation} 
where the last inequality holds due to $\vert \overline{S} \vert \le M-1$.

Proposition~\ref{pro:regretbound} can be obtained by applying (\ref{eq:prop2_key1}) and (\ref{eq:prop2_key2}) to the decomposition inequality (\ref{eq:decomposition}).
%{
%    \textbf{(2) When} $\hat{t}_h > h$ (i.e,. $\exists \setS$ such that $\nochosen_S(h) < l_h$) : In this case, we consider the set of non-near-optimal matchings that are sufficiently scheduled by $h$, and use it in the place of the whole set $\cnopts_{\vmweight}$ of non-near-optimal matching. The proof is similar to the case when $\hat{t}+h \le h$, and can be found in~\cite{tr}.
%}

}

\end{document}